\documentclass[preprint,times,a4paper]{article}
\usepackage{fullpage}

\usepackage{mathtools}
\usepackage{graphicx}
\usepackage{subcaption}

\usepackage{amssymb,amsmath}
\usepackage[pdftex]{hyperref}%

\usepackage{amsthm}

\newtheorem{proposition}{Proposition}

\newtheorem{theorem}{Theorem}
 \makeatletter
\newcommand{\neutralize}[1]{\expandafter\let\csname c@#1\endcsname\count@}
\makeatother

\newenvironment{theorembis}[1]
  {\renewcommand{\thetheorem}{\ref*{#1}$'$}%
   \neutralize{theorem}\phantomsection
   \begin{theorem}}
  {\end{theorem}} 

\newtheorem{corollary}{Corollary}

\newtheorem{lemma}{Lemma}

\newtheorem{claim}{Claim}

\theoremstyle{remark}
\newtheorem{remark}{Remark} 
\newtheorem{definition}{Definition}
\newtheorem{notation}{Notation}
\newtheorem{example}{Example}

\newenvironment{examplebis}[1]
  {\renewcommand{\theexample}{\ref*{#1}$'$}%
   \neutralize{example}\phantomsection
   \begin{example}}
 {\end{example}}

\newenvironment{examplerev}[1]
  {\renewcommand{\theexample}{\ref*{#1} \em revisited}%
   \neutralize{example}\phantomsection
   \begin{example}}
 {\end{example}}

\newenvironment{examplebis2}[1]
  {\renewcommand{\theexample}{\ref*{#1}$''$}%
   \neutralize{example}\phantomsection
   \begin{example}}
 {\end{example}}

\newtheorem*{problem*}{Open Problem}

\usepackage{float}
\usepackage{wrapfig}
\floatstyle{boxed}
\usepackage{tikz}
\usetikzlibrary{shapes,backgrounds}
\usetikzlibrary{patterns}
\usetikzlibrary{decorations.pathreplacing}
\usetikzlibrary{arrows,decorations.pathmorphing,backgrounds,fit,positioning,shapes.symbols,chains}

\newcommand{\poly}{\mathrm{poly}}
\newcommand{\CK}{\mathrm{C}}

\title{Resource-Bounded Kolmogorov Complexity Provides an Obstacle to Soficness of Multidimensional Shifts}

\begin{document}

\title{Resource-Bounded Kolmogorov Complexity Provides an Obstacle to Soficness of Multidimensional Shifts}

\author{Julien Destombes \and 
Andrei Romashchenko}

\maketitle

\begin{abstract}
We suggest necessary  conditions for soficness of multidimensional shifts formulated in terms of resource-bounded Kolmogorov complexity.
Using this technique we provide examples of effective and  non-sofic  shifts  on $\mathbb{Z}^2$ with very low block complexity:
the number of globally admissible patterns of size $n\times n$ grows only as a polynomial in $n$. We also show that more conventional proofs of non-soficness for multi-dimensional effective shifts,
including the techniques of  
Pavlov [Proc.\,AMS, 141(3):987-996, 2013] and
Kass and Madden [Proc.\,AMS, 141(11):3803-3816, 2013],
can be  expressed in terms of Kolmogorov complexity with unbounded computational resources.
 \end{abstract}

\section{Introduction}

Symbolic dynamics originally appeared in mathematics as a branch of the theory of dynamical systems that studies  smooth or topological dynamical systems by discretizing the underlying space. Since the late 1930s, symbolic  dynamics became an independent field of research, see \cite{symbolic-dynamics-1938,symbolic-dynamics-1940}. A classical \emph{dynamical system}  is a space  (of states) $\cal S$ with a function $F$ acting on this space; this function represents the ``evolution rule,'' i.e., the time dependence of a configuration in the space. The central notion of the theory of dynamical systems is a trajectory --- a sequence of configurations obtained by iterating  the evolution rules, 
\[
x, F(x), F(F(x)),\ldots, F^{(n)}(x),\ldots
\]

In symbolic dynamics the  space of states reduces to a finite set (an \emph{alphabet}). The trajectories are represented by  infinite (or bi-infinite) sequences of letters over this alphabet, and the ``evolution rule'' is the shift operator acting on these sequences.
Symbolic dynamics focuses on the   \emph{shift spaces} --- the sets of bi-infinite sequences of letters (over a finite alphabet)  that are defined by a shift-invariant constraint on the factors of finite length. More precisely,  a \emph{shift space} (also called a \emph{shift}) over an alphabet 
$\Sigma$ is a non-empty subset of bi-infinite sequences over $\Sigma$ that is
 translation invariant and closed in the standard topology of the Cantor space.
Every shift can be defined in terms of \emph{forbidden finite patterns}: we fix a set of (finite) words $\cal F$ and say that a configuration (a bi-infinite sequence)
belongs to the corresponding shift ${\cal S}_{\cal F}$ if and only if it does not contain any factor from~$\cal F$. We bear in mind that  different sets of forbidden patterns can induce the same shift.

Obviously, the properties of shifts heavily depend on the corresponding set of forbidden patterns. The following three large classes of shifts
play an important role in symbolic dynamics and  computability theory:
\begin{itemize}
\item \emph{shifts of finite type} (SFT), which can be defined by a finite set of forbidden finite patterns;
\item \emph{sofic} shifts (introduced in \cite{weiss}), where the set of forbidden finite patterns is a regular language;
\item \emph{effective} (or \emph{effectively closed} shifts),  which can be defined by a computable\footnote{We can require that the set of forbidden finite patterns is \emph{decidable} or \emph{recursively enumerable}. Though formally speaking these requirements are different, they result in one and the same class of shifts. That is, if some shift is defined by a recursively enumerable set of forbidden finite patterns, then the same shift can be also defined by a decidable set of forbidden patterns.}
 set of forbidden finite patterns.
\end{itemize}
The class of effective shifts is often understood as a class of ``explicitly constructible'' shifts, so it is of  considerable interest beyond the scope of computability theory.
These three classes are distinct:
\[
[\text{the SFTs}] \subsetneqq [\text{the sofic shifts}] \subsetneqq  [\text{the effective shifts}].
\]

The  sofic shifts  can be equivalently defined as the coordinate-wise projections of configurations from an SFT:
\begin{definition}\label{d:sofic}
A shift $\cal S$ over an alphabet $\Sigma$ is \emph{sofic} if there is an SFT ${\cal S}'$ over an alphabet $\Sigma'$ and a mapping 
$\pi : \Sigma' \to \Sigma$, such that  $\cal S$ consists of the coordinate-wise projections
\[
(\ldots \pi( y_{-1}) \pi (y_0 ) \pi(y_1) \pi(y_2)\ldots)
\]
of all configurations $(\ldots  y_{-1} y_0  y_1 y_2\ldots)$ from ${\cal S}'$.
\end{definition}
There is a simple characterization of soficness. Let us say that two words $w_1$, $w_2$ are \emph{equivalent} in a shift $\cal S$, if they have the same \emph{follower sets}, i.e.,  if exactly the same half-infinite configurations occur in $\cal S$ immediately to the right of $w_1$ and to the right of $w_2$. A shift is sofic if and only if the finite patterns in this shift are subdivided in a finite number of equivalence classes (see \cite[Theorem~3.2.10]{sft}). Loosely speaking, this criterion guarantees that when we read a configuration  from the left to the right and verify that it belongs to a sofic shift,
we need to keep in mind only finite information.

The SFTs and even the sofic shifts are  rather restrictive classes of shifts with several very special properties. Not surprisingly, many important examples of
effective shifts are not sofic. Non-soficness of a shift is usually proved with some version of the pumping lemma from automata theory.

\bigskip
\noindent \textbf{Multidimensional shifts.} The formalism of shifts can be naturally extended to the grids $\mathbb{Z}^d$ for $d>1$. A shift on $\mathbb{Z}^d$
(over a finite alphabet $\Sigma$) is defined as a set of $d$-dimensional configurations $f \ :\ \mathbb{Z}^d\to \Sigma$ that are (i)~translation-invariant 
(under translations in all directions) and (ii)~closed in Cantor's topology. Similar to the one-dimensional case, the shifts can be defined in terms of forbidden finite patterns. 

The definitions of the \emph{effective shifts} (the set of forbidden patterns is computable) and  of the \emph{SFTs} (the set of forbidden patterns is finite) 
apply to the  multidimensional shift spaces directly, without any revision. The \emph{sofic shifts} on  $\mathbb{Z}^d$ are defined as in Definition~\ref{d:sofic} above (as the coordinate-wise projections of SFTs). We refer the reader to  \cite{hochman2009dynamics} for an extensive discussion of the relations between SFT, sofic, and effective shifts on  $\mathbb{Z}^d$.

For multidimensional shifts spaces, the classes of the effective shifts, the sofic shifts, and the SFTs remain distinct, though the difference between these classes is more elusive than in the one-dimensional case. 
In this paper we discuss the tools that help to reveal the reasons why one or another effective multidimensional shift is \emph{not} sofic.

The class of sofic shifts in dimension $d\ge 2$ is surprisingly wide. Besides many simple and natural examples,  there are  shifts whose
soficness follow from rather subtle considerations.
For instance,  S.~Mozes showed that the shifts  generated by (a natural class of) non deterministic  multidimensional substitutions systems
 are sofic \cite{mozes}.
L.\,B.~Westrick proved that  
the two-dimensional shift on the alphabet $\{0,1\}$ whose configurations consist of  squares of $1$s of pairwise different sizes on a background of $0$s,
is sofic; moreover, any effectively closed subshift of this shift is also sofic   \cite{westrick}.

\begin{figure}
\begin{center}
  \begin{tikzpicture}[scale=0.27,x=1cm,baseline=2.125cm,ultra thick/.style= {line width=3.0pt},]
      \foreach \x in {0,...,15} 
    {
       \path[fill=red!41,draw=black] (\x,8) rectangle ++ (1,1);

    }

 \pgfmathsetseed{4}
    \foreach \x in {0,...,15} \foreach \y in {0,...,7}
    {
        \pgfmathparse{mod(int(random*23),2) ? "black!66" : "black!7"}
        \edef\colour{\pgfmathresult}
        \path[fill=\colour,draw=black] (\x,\y) rectangle ++ (1,1);
        \path[fill=\colour,draw=black] (\x,8+8-\y) rectangle ++ (1,1);
    }
    
    \draw [draw=blue,ultra thick]  (3,2) rectangle (7,6);
    \draw [draw=blue,ultra thick]  (3,11) rectangle (7,15);
       
\end{tikzpicture}
\caption{A configuration with  mirror symmetry with respect to the horizontal red line.
 The blue squares select  two symmetric black-and-white patterns.} 
\label{f:mirror}
\end{center}
\end{figure}

On the other hand, there are several examples of effective multidimensional shifts that are known to be non-sofic. In what follows we briefly discuss two of them.

\begin{example}[the mirror shift]
\label{ex:mirror}
One of the standard  examples of a non-sofic shift is the shift of mirror-symmetric configurations on $\mathbb{Z}^2$. Let $\Sigma$ be the alphabet with
three letters (e.g., \emph{black}, \emph{white}, and \emph{red}), and the configurations of the shift are all black-and-white configurations
(without any red letter) and the configurations with an infinite horizontal line of red letters and symmetric black-and-white half-planes 
above and below this line, see Fig.~\ref{f:mirror}. We denote this shift by ${\cal S}_{\text{mirror}}$.

It is easy to see that this shift is effective (the forbidden patterns are those where the red letters are not aligned, and those where the areas above and below the horizontal red line are not symmetric). At the same time, this shift is not sofic. The intuitive explanation of this fact is as follows. Let us focus on a pair of symmetric patterns of size $n\times n$ in black and white, above and below the horizontal red line (see the blue squares in Fig.~\ref{f:mirror}).  To make sure that the configuration belongs to the shift, we must ``compare'' these  two patterns with each other.
To this end, we need to transmit the information about a pattern of size $n^2$ through its border line (of length $O(n)$). However, in
a sofic shift, the ``information flow''  across a contour of length $O(n)$ is bounded by $O(n)$ bits, and this contradiction implies non-soficness. 
For a more formal argument see, e.g.  \cite{mirror} and \cite{jeandel-hdr}, or a similar example \cite[Example~2.4]{kass-madden}.
We revisit this example in Section~\ref{s:plain-epitomes} and use it as a simple illustration of our technique.
\end{example}

\begin{example}[the high complexity shift]
\label{ex:complexity}
Let $\cal S$ be the set of all binary configurations on $\mathbb{Z}^2$ where for each $n\times n$ pattern $P$ its Kolmogorov complexity is
quadratic, $C(P)=\Omega(n^2)$. Technically, this means that no globally admissible pattern can be produced by a program of size below  $c n^2$,
for some factor $c>0$ (see the formal definition of Kolmogorov complexity below).

This shift is not sofic. This follows from two facts (proven in \cite{dls}):
\begin{itemize}
\item[(i)] For some $c<1$, the shift defined above is not empty.
\item[(ii)] In every non-empty sofic shift on $\mathbb{Z}^2$, there is a configuration where the Kolmogorov complexity of each  $n\times n$ pattern  is bounded by $O(n)$.
\end{itemize}
At the same time, this shift is obviously effective: we can algorithmically enumerate the patterns whose Kolmogorov complexity is below the specified threshold. 
In Section~\ref{s:main} we study a similar construction based on a less conventional version of Kolmogorov complexity. 
\end{example}

Note that the non-sofic shifts in the two examples above have positive entropy (the number of globally admissible patterns of size $n\times n$
grows as $2^{\Omega(n^2)}$). This is not surprising: 
the proofs of non-soficness of these shifts use the intuition about the information flows
(\emph{super-linear amount of information cannot flow through a linear contour}).  
This type of argument can be adapted for some shifts where 
the number of globally admissible patterns of size $n\times n$ grows slower than  $2^{\Omega(n^2)}$ but still  faster than $2^{O(n)}$
(see, e.g. \cite[Proposition~15]{salo}).
As it was noticed in \cite{westrick}, 
``\emph{all examples known to the author of effectively closed shifts which are not sofic were obtained by in some sense allowing elements to pack 
too much important information into a small area.}''

This type of argument was formalized as rather general sufficient conditions of non-soficness in \cite{pavlov} and  \cite{kass-madden}. 
The theorems by  Kass and Madden (\cite[Theorem~3.2.10]{kass-madden}) and Pavlov (\cite[Theorem~1.1]{pavlov}) 
apply only to the two-dimensional shifts where the number of globally admissible $n\times n$ patterns is greater than $2^{O(n)}$. 
However, there is no reason to think that this condition is necessary for non-soficness (see, e.g. the discussion in \cite[Section~1.2.2]{jeandel-hdr}).
It is instructive to observe that non-effective non-sofic shifts can have very low block complexity, see  \cite{kass-madden,ormes-pavolov}.

In this paper we propose a new approach to the proof of non-soficness. 
Our technique uses  algorithmic descriptive complexity (Kolmogorov complexity). Though the definition of the sofic shifts does not involve explicitly any resource bounded computations, we show that a suitable tool to deal with soficness is the time-bounded version of Kolmogorov complexity.
We argue that a shift on $\mathbb{Z}^2$ cannot be sofic if the essential information contained in
an $n\times n$ pattern cannot be compressed to the size $O(n)$ in \emph{bounded time}. 

The ideas behind our argument are similar to those used in   \cite{dls} and \cite{kass-madden}, and
there is no surprise that Kolmogorov complexity helps to formalize the intuition of ``information flows.''
The main new element of our approach 
is the idea of  information compression with \emph{bounded computational resources}. 
This technique applies to several shifts with very low block complexity:
we cannot ``communicate'' the essential information  about a pattern across its contour not because this information is too large,
but since we do not have enough time and space to compress it.
In particular, we provide examples of non-sofic effective shifts with  only polynomial block complexity (and thus zero entropy).

\begin{remark}\label{rem:1}
A more straightforward  approach to the measure of the ``\emph{information flow through the border line of a pattern}'' might use the notion of \emph{extender sets} introduced in  \cite{kass-madden}. Extender sets are a  natural generalization of follower sets to multi-dimensional shifts. Let $\cal S$ be a shift and $P$ be a globally admissible finite pattern in this shift. The extender set of $P$ in $\cal S$ is the set of all (co-finite) patterns $Q$ completing $P$ to a valid configuration of $\cal S$  (the support of $Q$ should be the complement of the support of $P$). 

Let $\cal S$ be a shift on $\mathbb{Z}^2$;  denote by $N_k$ the number of different extender sets   for the globally admissible patterns with a support of size $k\times k$.  (Several patterns can share one and the same extender set, so the number of extender sets might be much less than the number of all globally admissible patterns of this size). It seems natural to interpret $\log N_k$ as  the ``{information flow}'' crossing the border line of a $k\times k$ pattern and imposing the consistency  between the pattern and its complement.  As the length of the contour around a $k\times k$ pattern is $O(k)$, the naive  intuition suggests that ``a super-linear flow of information cannot cross a contour of linear length.'' This is indeed true for shifts of finite type: for every SFT on $\mathbb{Z}^2$, the number $N_k$ is not bigger than $2^{O(k)}$. However, in case of sofic shifts, this approach fails. For a sofic shift on $\mathbb{Z}^2$, the value of  $\log N_k$ can grow much faster than the length of the border line of the pattern (\cite{kass-madden} attributes this observation  to unpublished works of C.~Hoffman,  A.~Quas, and R.~Pavlov). In fact, for a  sofic shift on $\mathbb{Z}^2$,   the value of $\log N_k$ can grow even as $\Omega(k^2)$ (we mention such an example in Section~\ref{s:ordered-epitomes}). Therefore, we cannot use the asymptotic of $\log N_k$ to prove non-soficness of a multi-dimensional shift. 
This is why we need a subtler  refinement  of the  intuition of ``information flows'' in sofic shifts.

Though the \emph{cardinality} of the family of extender sets  is not enough to distinguish between sofic and non-sofic shifts, 
some finer properties of extender sets can be helpful.
Kass and Madden proposed in \cite{kass-madden} an  approach  to the proof of non-soficness based on the structure of these sets. 
Technically, they focus on the length of union-increasing chains of extender sets, which is a more appropriate parameter than the total number of these sets.
We discuss this method in Section~\ref{s:km} and show that it is can be interpreted as a special case of  our approach.
\end{remark}

The rest of the paper is organized as follows.
After recalling the  necessary definitions from symbolic dynamics and the theory of Kolmogorov complexity in Section~2, we prove in Section~3 our main technical results and construct a non-sofic effective shift with very low block complexity. 
In Section~4 we extend our technique to more general  settings; we discuss several examples and 
show that every proof of non-soficness  using the argument from  \cite{kass-madden} 
(the method of union-increasing chains of extender sets) 
can be reformulated in terms of our technique  with the resource-unbounded  Kolmogorov complexity.
For the sake of self-containedness, in Appendix we prove   Theorem~\ref{th:dls} (implicitly proven in \cite{dls}) and  Lemma~\ref{l:kass-madden} (formulated and proven in different terms in \cite{kass-madden}). We also prove in Appendix  Theorem~\ref{th:deep-shift}, which is a stronger version of Theorem~\ref{th:deep-shift-simple} (though the simpler Theorem~\ref{th:deep-shift-simple} proven in the main text of the paper is strong enough for all applications where we use it to construct effective non-sofic shifts).

\section{Preliminaries}

\paragraph{Shift spaces}
In the introduction  we have defined the notion of \emph{shift spaces} (also called \emph{shifts}). More specifically, we study \emph{shifts of finite type}, \emph{sofic shifts}, and \emph{effective shifts}. In this paper we focus on two-dimensional shifts, though all arguments can be extended to the shifts on  $\mathbb{Z}^d$ for all $d\ge 2$.

A (finite) \emph{pattern} on $\mathbb{Z}^2$ over a finite alphabet $\Sigma$ is a mapping from a (finite) subset of $\mathbb{Z}^2$ to $\Sigma$;
the domain of this mapping is the \emph{support} of the pattern. Sometimes a pattern $P$ with a support $\cal A$ is called a \emph{coloring}
of $\cal A$ (the ``colors'' are letter from $\Sigma$).

For a shift $\cal S$, we say that a pattern $P$ is \emph{globally admissible}, if $P$ is a restriction of a configuration from $\cal S$ onto some finite support.
For a shift of finite type determined by a set of forbidden patterns $\cal F$, we say that  a pattern is   \emph{locally admissible} if it contains no forbidden  patterns from $\cal F$. Similarly, two patterns with disjoint domains are \emph{locally compatible} if the union of these patterns is locally admissible.

The \emph{block complexity} of a shift is a function that gives for each integer $n>0$ the number of globally admissible patterns of size $n\times n$ 
(patterns with support $\{1,\ldots,n\}^2$) in this shift. 

If a sofic shift $\cal S$ is a coordinate-wise projection of configurations from $\hat {\cal S}$, we say that $\hat {\cal S}$ is a \emph{covering} of $\cal S$. Every sofic shift has a covering SFT that can be defined by a set of forbidden patterns $F$, where every pattern in $F$ is a pair of neighboring letters (see, e.g.  \cite{sft}).

\paragraph{Computability and Kolmogorov complexity}
In this section we recall the main  definitions of the theory of Kolmogorov complexity.
In the way that is usual in computability theory (see, e.g., \cite{nies2009computability,rogers1987theory}), we identify \emph{algorithms} with Turing machines. A partial function $f:\{0,1\}^*\to \{0,1\}^*$ is called  \emph{computable}, if there is an algorithm (a Turing machine) $\cal M$ that computes this function in the following sense:   given as the input a string $x$ from the domain of $f$,
machine $\cal M$ prints the value $f(x)$ and stops; given as the input a string $x$ such that $f(x)$ is not defined, machine $\cal M$ never stops. It is known that the class of computable functions does not change if we vary minor details of the definition of a Turing machine (we may admit machine with only one tape or  with any number of tapes, the tapes can be bi-infinite or only infinite to the right, and so on).
It is generally agreed that this definition is an adequate formalization of the intuitive idea of computability. In this paper we discuss several simple algorithms;  we give only an informal description of these algorithms and do not construct explicitly the corresponding Turing machines.

Let $U$ be a (partial) computable function. The  complexity of $x$ with respect to the description method $U$ is defined as
\[
\CK_U(x) :=  \min \{  \ |p| \ : \ U(p) = x\ \},
\]
where $|p|$ stands for the length of the binary string $p$. 
If there is no $p$ such that  $U(p) = x$, we assume that $\CK_U(x)=\infty$.
Here $U$ is understood as  a programming language; $p$ is a program that  prints  $x$; the complexity of $x$  is the length of (one of) the shortest programs $p$ that generate $x$.

The obvious problem with this definition is its dependence on $U$. The theory of Kolmogorov complexity becomes possible due to the
following invariance theorem.
\begin{theorem}[Kolmogorov, \cite{kolmogorov-three-approaches}]

\label{t:kolmogorov-invariance}
There exists a computable function  $U$ such that for any other computable function $V$ there is a constant $d$ such that
\[
\CK_U(x) \le \CK_V(x)+d
\]
for all $x$. 
\end{theorem}
\begin{notation}
This $U$ is called \emph{an optimal description method}. We fix an optimal $U$ and in what follows omit the subscript in $\CK_U(x)$ and use the notation $\CK(x)$.
The value $\CK(x)$ is called the (plain)  Kolmogorov complexity of $x$.
\end{notation}

In a similar way, we  define Kolmogorov complexity in terms of programs with bounded resources (the time of computation). Let $U$ be a Turing machine;
we define the Kolmogorov complexity $\CK^{t}_{U}(x)$ as the
length of the shortest $p$ such that $U(p)$ produces $x$
in at most $t$ steps.  There exists an \emph{optimal description method} $U$ in the following sense:
for every Turing machine $V$ 
\[
\CK_{U}^{\poly(t)}(x)\le \CK_{V}^{t}(x)+O(1).
\]
For multi-tape Turing machines a slightly stronger statement can be proven:
\begin{theorem}[see~\cite{li-vitanyi-book}; the proof uses the simulation technique from  \cite{hennie-stearns}]

There exists an \emph{optimal description method} \textup(multi-tape Turing machine\textup) $U_m$ in the following sense:
for every multi-tape Turing machine $V$ \textup(with any number of tapes\textup) there exists a constant $d$ such that 
\[
\CK_{U_m}^{ct\log t}(x)\le \CK_{V}^{t}(x)+d
\]
 for all strings $x$.     
\end{theorem}

\begin{notation}
We fix such a machine $U_m$, and in the sequel use for the resource-bounded version of Kolmogorov complexity
the notation $\CK^{t}(x)$ instead of $\CK_{U_m}^{t}(x)$.
Without loss of generality we may assume that $\CK(x)\le \CK^{t}(x)$ for all $x$ and for all~$t$.
\end{notation}

It is easy to show that there exists a constant $d$ such that for all bit strings $x=(x_1\ldots x_n)\in\{0,1\}^n$ we have $\CK(x)\le n +d$
(roughly speaking, this is true because there always exists a trivial algorithm of type  ``\emph{print$(x_1\ldots x_n)$}'' where the bits $x_i$ are listed explicitly; such an algorithm can be expressed in any natural programming language, including Turing machines). At the same time, for every $n$ there exists a string $x \in\{0,1\}^n$ such that  $\CK(x)\ge n$,  because the total number of descriptions $p$ that are shorter than $n$ bits, is less than $2^n$. (Such a string $x$ is usually called \emph{incompressible}.)
More generally, in every set of $2^n$ strings there is at least one string $x$ such that $\CK(x)\ge n$.

We fix a computable enumeration of finite patterns (over a finite alphabet) that assigns a binary string (a \emph{code}) to each pattern in dimension two. 
In the sequel we take the liberty of talking about Kolmogorov complexity of finite patterns in dimension two 
(assuming the Kolmogorov complexity of the \emph{codes} of these patterns). Similarly, we can talk about Kolmogorov complexity of natural numbers (assuming the Kolmogorov complexity of the binary expansion of these numbers).

In computability theory it is common to study \emph{relativized computations}, i.e., computations with an oracle. Let $\cal O$ be an infinite binary sequence. A Turing machine  \emph{with oracle} $\cal O$ is a Turing machine with an extra read-only tape where the bits of $\cal O$ are written; the machine can read bits of the oracle when necessary and then use them in the computation on the usual read/write working tape. Though the access to an oracle can extend the power of a Turing machine, many fundamental properties of Turing machines remain valid with any oracle (in particular, the arguments based on diagonalization apply to Turing machines with an oracle). We refer the reader to the  textbooks \cite{nies2009computability,rogers1987theory} for a detailed discussion of relativized computations.

\begin{notation}
Having fixed an infinite sequence $\cal O$, we can define the relativized version of Kolmogorov complexity for machines with oracle $\cal O$ (in exactly the same way as above but with Turing machines having the access to the chosen oracle). We use the notation $\CK^{\cal O}(x)$ for Kolmogorov complexity with oracle $\cal O$; similarly, $\CK^{t,{\cal O}}(x)$ denotes Kolmogorov complexity with oracle $\cal O$ and time bound $t$.
\end{notation}

A systematic introduction to the theory of Kolmogorov can be found in \cite{li-vitanyi-book,shen-vereshchagin-book}.
In what follows we use only rather simple properties of Kolmogorov complexity. Throughout the paper,  the claims of the form $\CK^{t}(x)=m$ can be understood as an informal statement \emph{there is a program of length $m$ that prints $x$ in $t$ steps}; we avoid  statements where subtle details of the formal definition of Kolmogorov complexity become significant.

\section{High resource-bounded Kolmogorov complexity is compatible with low block complexity}
\label{s:main}

The following theorem was proven implicitly in \cite{dls}:
\begin{theorem}\label{th:dls}
In every non-empty sofic shift $\cal S$ on $\mathbb{Z}^2$ there exists a configuration $x$ such that for all $n\times n$-patterns $P$ in $x$, we have
\[
\CK^{T(n)}(P) = O(n)
\]
for a time threshold  $T(n)=2^{O(n^2)}$. 
\end{theorem}
In \cite{dls} a weaker version of this theorem is stated:  it is claimed only that the \emph{plain} complexity of $n\times n$ patterns is $O(n)$.
However, the argument from \cite{dls} implies a bound for a \emph{resource-bounded} version of Kolmogorov complexity. 
For the sake of self-containedness, in what follows we provide a proof of this theorem in \ref{appendix-dls}.

Theorem~\ref{th:dls} guarantees that a shift is non-sofic, if in every configuration we can find a sequence of growing size patterns with super-linear Kolmogorov complexity. In the next theorem we prove the existence of such effective shifts.

\begin{theorem}\label{th:deep-shift}
For every $\epsilon>0$ and for every computable $T(n)$ there exists an effective shift on $\mathbb{Z}^2$ such that for every $n$ and 
for every globally admissible pattern $P$ of size $n\times n$, we have that
 \begin{itemize}
  \item[(i)] $\CK(P) = O(\log n)$, and
  \item[(ii)] $\CK^{T(n)}(P) = \Omega(n^{2-\epsilon})$.
 \end{itemize}
\end{theorem}
We defer the proof of Theorem~\ref{th:deep-shift} to  \ref{appendix-deep-shift}. In what follows we prove a slightly weaker version of this theorem,
which is nevertheless strong enough for our main applications:
\begin{theorembis}{th:deep-shift}
\label{th:deep-shift-simple}
For every computable $T(n)$ there exists an effective shift on $\mathbb{Z}^2$ such that 
 \begin{itemize}
  \item[(i)]  for every $n$ and for every globally admissible pattern $P$ of size $n\times n$, we have 
    $\CK(P) = O(\log n),$ and 
    \item[(ii)] for infinitely many $n$ and  for every globally admissible pattern $P$ of size $n\times n$, we have that
    $\CK^{T(n)}(P) = \Omega(n^{1.5}).$
 \end{itemize}
\end{theorembis}

From Theorem~\ref{th:dls} and Theorem~\ref{th:deep-shift-simple} we deduce the following corollary:
\begin{corollary}
There exists an effective non-sofic shift on $\mathbb{Z}^2$ with block complexity $\poly(n)$,
i.e., with $\le\poly(n)$ globally admissible blocks of size $n\times n$.
\end{corollary}
\begin{proof}
We take the shift from Theorem~\ref{th:deep-shift-simple} assuming that the threshold $T(n)$ is much greater than $2^{\Omega(n^2)}$ (e.g., we can let $T(n)=2^{n^3}$). On the one hand, property~(ii) of Theorem~\ref{th:deep-shift-simple} and Theorem~\ref{th:dls}
guarantee that this shift is not sofic. On the other hand, property~(i) of Theorem~\ref{th:deep-shift-simple} implies that the number of globally admissible
blocks of size $n\times n$ is not greater than $2^{O(\log n)}$.
\end{proof}
\begin{remark}\label{r:improved-upper-bound}
Our proof of Theorem~\ref{th:deep-shift-simple} implies a stronger bound than property~(i).  In fact, instead of the bound
$\CK(P) = O(\log n)$
we can prove for  some (large enough) computable function $\hat T(n)$ that 
for every globally admissible $n\times n$ pattern $P$ in this shift we have
\begin{equation}
\CK^{\hat T(n)}(P) =O( \log n). \nonumber 
\end{equation}
\end{remark}

Now we introduce some notation.
\begin{notation}\label{d:paradoxical-shift}
For any function $T(n)$ and for any number $\lambda>0$, let  ${\cal S}_{\lambda,T}$ denote the following  shift over the alphabet $\{0,1\}$:  a configuration $x$ belongs  to ${\cal S}_{\lambda,T}$  if and only if $\CK^{T(n)}(P) \le \lambda \log n$ for every $n\times n$ pattern $P$ in $x$. 
\end{notation}

Obviously, for an integer $\lambda$ and a computable $T(n)$, the shift introduced in Notation~\ref{d:paradoxical-shift} is effective.

Remark~\ref{r:improved-upper-bound} can be rephrased as follows: there are
an integer number $\hat \lambda$ and a computable function $\hat T(n)$ such that
for the corresponding shift ${\cal S}_{\hat \lambda,\hat T}$
there exists a sequence of  globally admissible patterns $P_n$ of size $n\times n$ ($n=1,2,\ldots$) 
such that 
\[
\CK^{2^{n^3}}(P_n) = \Omega (n^{1.5});
\] 
the same time, for every globally admissible pattern of size $n\times n$
we have $\CK(P) \le \CK^{\hat T(n)}(P) =O( \log n)$.
\begin{notation}
\label{d:s_0}
Let ${\cal S}_{0}$ denote the shift defined above.
\end{notation}

The shift constructed in  Theorem~\ref{th:deep-shift-simple} is a proper subshift of ${\cal S}_{0}$.  Indeed, besides all configurations of the shift from Theorem~\ref{th:deep-shift-simple} (those configuration must have only patterns with high resource-bounded Kolmogorov complexity),  ${\cal S}_{0}$ admits  also  patterns with a very low time bounded complexity (e.g., the configuration with all $0$s and the configuration with all $1$s).
In the next section we use  ${\cal S}_{0}$  to construct some other examples of effective non-sofic shifts.

\begin{proof}[(Proof of Theorem~\ref{th:deep-shift-simple})]
In this proof we provide a construction of a required shift  assuming that we have fixed in advance a universal Turing machine in the definition of Kolmogorov complexity.
The  choice of a universal Turing machine  is not unique (and quite arbitrary), and  the final choice of parameters in our construction depends  on this arbitrariness.

Let us fix a sequence  $(n_i)$ where  $n_0$ is a large enough integer number\footnote{For every instance of an optimal Turing machine $U$ in the definition of Kolmogorov complexity one can compute the corresponding minimum value of  $n_0$ that makes the construction work.}, and 
\begin{equation}
\label{eq:th2-c}
 n_{i+1}: =(n_0\cdot \ldots \cdot  n_i)^c\ \mbox{ for } i=0,1,2,\ldots,
\end{equation}
where $c\ge 3$ is a constant. 
We set $N_i := n_0 \cdot \ldots \cdot n_i$.
In what follows we construct for each $i$ a pair of  \emph{standard} binary patterns $Q_i^0$  and $Q^1_i$ of size $N_i \times N_i$  such that
  \begin{itemize}
  \item the plain Kolmogorov complexities of the standard patterns $\CK(Q_i^0)$ and $\CK(Q_i^1)$ are not greater than $O(\log N_i)$, and
  \item the resource-bounded Kolmogorov complexities $\CK^{T(N_i)}(Q_i^0)$ and $\CK^{T(N_i)}(Q_i^1)$ are not less than $\Omega(N_i^{1.5})$.
  \end{itemize}
The construction is hierarchical:   both $Q_i^0$  and $Q^1_i$ are defined as $n_i\times n_i$ matrices composed of  patterns 
$Q_{i-1}^0$  and $Q^1_{i-1}$; for each $i$ the blocks $Q_i^0$  and $Q^1_i$  are bitwise inversions of each other.

When the standard patterns $Q_i^0$  and $Q^1_i$  are constructed for all $i$, 
we define the shift as the closure of these patterns: we say that a finite pattern  is globally admissible if and only if it appears 
in some standard pattern  $Q_i^j$  or at least in a $2\times 2$-block composed of $Q_i^0$  and $Q^1_i$  (for some $i$).
\begin{remark}\label{r:low-level}
Due to the hierarchical structure of the standard patterns, 
we can guarantee that every globally admissible  pattern $P$ of size $N_i\times N_i$ appears   
in a $2\times 2$-block composed of $Q_i^0$  and $Q^1_i$ 
(no need to try the blocks $Q_{s}^j$ for $s>i$).
\end{remark}   
Since the construction of   $Q_i^j$ is explicit, the resulting shift is effective. Properties~(i) and~(ii) of the theorem will follow from the properties of the standard patterns.
   
 In what follows we explain an inductive construction of $Q_i^0$  and $Q^1_i$. Let $Q_0^0$  and $Q^0_1$ be the squares composed of only  $0$s and only $1$s respectively. Further, for every $i$ we take the lexicographically first binary  matrix $R_i$   of size $n_i\times n_i$ such that 
\begin{equation}
\label{eq:incompressible}
  \CK^{t_i}(R_i) \ge n_i^2
\end{equation}
(the time bound $t_i$ is fixed in the sequel).
We claim  that such a matrix exists. Indeed,  there exists a matrix of size $n_i\times n_i$ that is incompressible in the sense of the plain
Kolmogorov complexity. The resource-bounded Kolmogorov complexity of a matrix can be only greater than the plain complexity.
Therefore,  there exists at least one matrix satisfying \eqref{eq:incompressible}.
If $t_i$ is a computable function of $i$, then given $i$ we can find $R_i$  algorithmically.

Now we substitute in $R_i$ instead of each zero and one entry the copies of   $Q_{i-1}^0$  and $Q^1_{i-1}$  respectively, e.g., 
\[
R_i=\left(
\begin{array}{cccccc}
0&0&0&0&1\\
0&1&0&0&1\\
1&1&1&1&0\\
0&1&1&0&0\\
0&1&0&1&0
\end{array}
\right) 
\ \Longrightarrow\ 
Q_i^0:=\left(
\begin{array}{c|c|c|c|c}
Q_{i-1}^0&Q_{i-1}^0&Q_{i-1}^0&Q_{i-1}^0&Q_{i-1}^1\\
\hline
Q_{i-1}^0&Q_{i-1}^1&Q_{i-1}^0&Q_{i-1}^0&Q_{i-1}^1\\
\hline
Q_{i-1}^1&Q_{i-1}^1&Q_{i-1}^1&Q_{i-1}^1&Q_{i-1}^0\\
\hline
Q_{i-1}^0&Q_{i-1}^1&Q_{i-1}^1&Q_{i-1}^0&Q_{i-1}^0\\
\hline
Q_{i-1}^0&Q_{i-1}^1&Q_{i-1}^0&Q_{i-1}^1&Q_{i-1}^0
\end{array}
\right) 
\]

The resulting matrix (of size $N_i\times N_i$) is denoted  $Q_i^0$. Matrix $Q_i^1$ is defined as the bitwise inversion of $Q_i^0$.

\smallskip
\noindent
\emph{Claim 1}. Assuming that $t'_i\ll t_i$  (in what follows we discuss the choice of  $t_i'$ in more detail) we have
\[
  \CK^{t'_i}(Q_i^0)  = \Omega(N_i^{1.5}) \mbox{ and }   \CK^{t'_i}(Q_i^1)  = \Omega(N_i^{1.5}). 
\]

\begin{proof}[Proof of Claim~1]
Given $Q_i^j $   (for $j=0,1$) we can  retrieve the matrix $R_i$ (this retrieval can be implemented in polynomial time). Therefore, for every time bound $t$ 
\[
\CK^{t+ \poly(N_i)} (R_i) \le \CK^{t}(Q_i^j) + O(1).
\]

Therefore, if $t_i > t_i' + \poly(N_i)$ then
\[
 n_i^2 \le \CK^{t_i}(R_i) \le  \CK^{t_i'}(Q_i^j). 
\]
It remains to observe that  our choice of parameters in \eqref{eq:th2-c} with $c\ge 3$ implies
$ n_i^{1/2} \geq \left(n_0 \cdot \ldots \cdot n_{i-1} \right)^{3/2}$, and therefore
\[
n_i^2 \geq \left(n_0 \cdot \ldots \cdot n_i\right)^{1.5} = (N_i)^{1.5}.
\]
Thus, we obtain $\CK^{t'_i}(Q_i^j)  \ge  (N_i)^{1.5} - O(1)$, and the claim is proven.
\end{proof}
\begin{remark}
By choosing a larger constant $c$ in \eqref{eq:th2-c}, we can achieve a lower bound $  \CK^{t'_i}(Q_i^j) = \Omega(n^{2-\epsilon})$ for any $\epsilon>0$.
\end{remark}
\smallskip
\noindent
\emph{Claim 2}. For every globally admissible pattern $P$ of size $N_i\times N_i$ 
(and not only for the standard patterns, as it was in Claim 1)
 its time-bounded Kolmogorov complexity $\CK^{T(N_i)}(P) $  is $\Omega(n^{1.5})$ (assuming that  $T(N_i) \ll t_i'$).

\smallskip

\begin{figure}
\begin{center}

  \begin{tikzpicture}[scale=0.27,x=1cm,baseline=2.125cm]

  
    \draw [fill=black!30, thick, dashed]  (3,6) rectangle (11,14);
 
    \foreach \x in {0,...,15} \foreach \y in {0,...,15}
    {
        \path[draw=black] (\x,\y) rectangle ++ (1,1);
    }

      \draw [ultra thick]  (0,0) rectangle (8,8);
      \draw [ultra thick]  (0,8) rectangle (8,16);
      \draw [ultra thick]  (8,0) rectangle (16,8);
      \draw [ultra thick]  (8,8) rectangle (16,16);


 \begin{scope}[shift={(20,0)}]
  
    \draw [fill=black!30, thick, dashed]  (3,6) rectangle (11,14);

    \draw [fill=red!30]  (3,6) rectangle (8,8);
    \draw [pattern=crosshatch, pattern color=red!91]  (11,6) rectangle (16,8);

    \draw [fill=blue!30]  (3,8) rectangle (8,14);
    \draw [pattern=north east lines, pattern color=blue!90]  (11,0) rectangle (16,6);

    \draw [fill=orange!30]  (8,8) rectangle (11,14);
    \draw [pattern=north west lines, pattern color=orange!90]  (8,0) rectangle (11,6);

    \foreach \x in {0,...,15} \foreach \y in {0,...,15}
    {
        \path[draw=black] (\x,\y) rectangle ++ (1,1);
    }

      \draw [ultra thick]  (0,0) rectangle (8,8);
      \draw [ultra thick]  (0,8) rectangle (8,16);
      \draw [ultra thick]  (8,0) rectangle (16,8);
      \draw [ultra thick]  (8,8) rectangle (16,16);

   \draw[-latex, ultra thick,black,dashed](6.5,12.5)node{ }  to[out=0,in=90] (14.5,2.5);
   \draw[-latex, ultra thick,black,dashed](9.1,9.5)node{ }  to[out=-135,in=135] (9.1,1.5);
   \draw[-latex, ultra thick,black,dashed](4.5,7.0)node{ }  to[out=-90,in=-135] (12.5,6.5);

\end{scope}
  
   \node at (8.0,-2.0) {(a)};
   \node at (28.0,-2.0) {(b)};

\end{tikzpicture}

\caption{A pattern of size $N_k\times N_k$ (shown in gray in fig.~(a)) covered by a quadruple of standard blocks of the same size
contains enough information to reconstruct a standard pattern (fig.~(b)).} \label{f:non-standard-block}
\end{center}
\end{figure}

\begin{proof}[Proof of Claim~2]
If a pattern $P$ of size $N_i\times N_i$ is globally admissible then it is covered by a quadruple of standard patterns of rank $i$, see Remark~\ref{r:low-level} above.
Then $P$ can be divided into four rectangles which are ``corners'' of standard patterns of rank $i$, see Fig.~\ref{f:non-standard-block}~(a). 
Since the standard blocks  $Q_i^0$   and $Q_i^1$  are the inversions of each other, these four ``corners''  (with a bitwise inversion if necessary)  form together the entire standard pattern, as shown in Fig.~\ref{f:non-standard-block}~(b).
Therefore, we can reconstruct $Q_i^j$ from $P$ given (a)~the position of $P$ with respect to the grid of standard blocks (this involves $O(\log N_i)$ bits) 
and (b)~the four bits identifying the standard blocks covering $P$ (we need to know which of them is a copy of $Q_i^0$   and which one is a copy of $Q_i^1$).

The retrieval of  $Q_i^j$ from $P$  requires only $\poly(N_i)$ steps of computation (in addition to the time we need to produce $P$).
Now the claim follows from the bound for the resource-bounded Kolmogorov complexity of the standard patterns $Q_i^0$   and $Q_i^1$.
\end{proof}

\smallskip
\noindent
\emph{Claim 3}. For every $k\times k$-pattern in $Q_i^0$ or $Q_i^1$, its plain Kolmogorov complexity is at most $O(\log k)$.
\begin{proof}[Proof of Claim~3]
First of all, we observe that the standard patterns $Q_i^0$ or $Q_i^1$ can be computed given $i$. Therefore, 
  $
  \CK(Q_i^0) = O(\log i)\mbox{ and } \CK(Q_i^1) = O(\log i).
  $

Every globally admissible $k\times k$-pattern is covered by at most four  standard patterns $Q_i^0$ or $Q_i^1$ with
\[
N_{i-1} <k \le N_{i},
\]
see Remark~\ref{r:low-level}. Therefore, to obtain a globally admissible pattern $P$ of size $k\times k$ we need to produce a quadruple of standard patterns
of size $N_{i}\times N_i$ and then to specify the position of $P$ with respect to the grid of standard blocks. This description consists of only
 $
 O(\log N_i)
 $
bits, and we conclude that 
 $
 \CK(P) = O(\log k).
 $
\end{proof}
It remains to fix the time bounds $t_i$.
Given a computable threshold $T(N_i)$, we choose a suitable $t_i'\gg T(N_i)$ and then a suitable $t_i\gg t_i'$.
The theorem follows from Claim~2 and Claim~3.
\end{proof}

\begin{remark}
For all large enough $i$, the incompressible pattern $R_i$ constructed in the proof of Theorem~\ref{th:deep-shift-simple} 
contains copies of all $2^4$ binary patterns of size $2\times 2$.
Therefore, we can guarantee that every standard block $Q_i^j$ contains all globally admissible patterns of size $N_{i-1}\times N_{i-1}$. 
\end{remark}

\label{s:no-go}

There exists a non-empty effective shift on $\mathbb{Z}^2$ where the Kolmogorov complexity of
all $n\times n$ patterns is $\Omega(n^2)$ (see \cite{dls} and \cite{rumyantsev-ushakov}). So a natural question arises: 
can we improve  Theorem~\ref{th:deep-shift-simple} and strengthen condition~(ii) to
$\CK^{T(n)}(P) = \Omega(n^{2})$? The answer is negative: we cannot achieve  the resource bounded complexity $\Omega(n^{2})$,
even with a much weaker version of property~(i) for the plain complexity:
\begin{proposition}\label{p:non-quadratic}
For all large enough time bounds $T(n)$, there is no shift on $\mathbb{Z}^2$ such that 
 \begin{itemize}
  \item[(i)]  
  for every globally admissible pattern $P$ of size $n\times n$, we have that $\CK(P) = o(n^2)$, and 
    \item[(ii)] 
   for infinitely many $n$ and  for every globally admissible pattern $P$ of size $n\times n$,  we have that  $\CK^{T(n)}(P) = \Omega(n^{2}).$
 \end{itemize}
\end{proposition}
\begin{proof}
Assume for the sake of contradiction that such a shift exists. Let us fix a real number $\varepsilon>0$.
Condition~(i) means that there exists an integer  $k=k(\varepsilon)$ such that
the number of globally admissible $k\times k$ patterns in this shift  is 
$
 L_k \le 2^{\varepsilon k^2}. 
$

Let us estimate Kolmogorov complexity of globally admissible patterns of size $(Nk) \times (Nk)$.
For any $N$, every globally admissible  pattern $P$ of size $(Nk)\times (Nk)$ can be specified by
 \begin{itemize}
 \item the list of all globally admissible patterns of size $k\times k$, which requires $L_k \cdot k^2$ bits (with  $O(L_k \cdot k^2$) bits we can even provide a self-delimiting description of this list), and
 \item by an array of size $N\times N$ with indices of $k\times k$ blocks  that constitute $P$ (which requires $N^2 \cdot \log L_k$ bits).
 \end{itemize}
We can combine these two parts of the description in one program. Clearly, $P$ can be reconstructed from such a description in polynomial time.  Thus, it follows that for some polynomial $\poly_1(n)$
\begin{eqnarray}
\CK^{\poly_1(Nk)} (P) \le  O(L_k \cdot k^2) + N^2 \cdot \epsilon k^2. \label{eq:Nk}
\end{eqnarray}

Now we estimate Kolmogorov complexity of a globally admissible pattern $P$ of size $n\times n$ in case when $n$ is not divisible by $k$. 
Let $N$ be the integer number such that 
\[
(N-1)k  < n \le N k.
\]
Pattern $P$ can be embedded in a bigger globally admissible pattern $P'$ of size $ (Nk)\times  (Nk)$.  To describe $P$, we need to describe $P'$ and then specify the position of $P$ with respect to the borders of $P'$.
The gaps between $P$ and the borders of $P'$ cannot be bigger than $k$. Therefore,
to obtain an upper bound for Kolmogorov complexity of $P$, we only need to add $O(\log k)$ bits to \eqref{eq:Nk}.
It follows that  for some polynomial $\poly_2(n)$
\begin{eqnarray}
\CK^{\poly_2(n)} (P)& \le&
O(L_k \cdot k^2) + N^2 \cdot \varepsilon k^2 + O(\log k) \nonumber\\
&\le& \epsilon (n+k)^2 + O(L_k \cdot k^2) \nonumber\\
&\le& \epsilon n^2 + 2 \varepsilon nk+ \varepsilon k^2  + O(L_k \cdot k^2). \label{eq:n^2}
\end{eqnarray}
 For  a fixed $k$ and large enough $n$, the first term in \eqref{eq:n^2} dominates the remaining terms, and we get 
\[
\CK^{\poly_2(n)} (P) \le   2\epsilon n^2. 
\]
 Let us recall that we can prove this bound for every $\varepsilon>0$ and for all large enough $n$.
 This contradicts condition~(ii) of the proposition.
\end{proof}

\section{Epitomes}
\label{s:epitomes}
The technique from Section~\ref{s:main} does not apply to the shifts that contain very simple configurations (with low resource-bounded
Kolmogorov complexity of all patterns). In particular, it does not apply to Example~\ref{ex:mirror} from the introduction.
In this section we propose a  different technique (also based on Kolmogorov complexity) that helps to handle such shifts.
The intuitive idea behind this technique is as follows: we try to capture the ``essential'' information in each pattern (discarding irrelevant data) and then measure the plain or resource-bounded Kolmogorov complexity of an ``epitome'' extracted from this essential information.
We show that  in a sofic shift  the complexity of properly defined ``epitomes of the essential information''   must be bounded by the length of the contour around the pattern. Thus, if the this condition is violated, then the shift is not sofic.

Let us fix some notation.
We denote by $B_n$ the set $\{0,\ldots, n-1\}^2\subset \mathbb{Z}^2$ and  by $F_n$ its complement, $F_n:=\mathbb{Z}^2\setminus B_n$. We say that
two patterns with disjoint supports are \emph{compatible} (for a shift $\cal S$) if the union of these patterns is globally admissible in $\cal S$. 
In particular, a finite pattern $P$ with support $B_n$ and an  infinite pattern $R$ with support $F_n$ are compatible, if the union of these patterns is a valid configuration of the shift. 

In Section~\ref{s:plain-epitomes} we start with a simple and more restrictive definition of \emph{plain epitomes}. Then, in Section~\ref{s:ordered-epitomes} we proceed with a more general definition of \emph{ordered epitomes} and discuss several examples of proofs of non-soficness based on this technique. In Section~\ref{s:km} we compare the technique of epitomes with the technique of \emph{union-increasing chains of extender sets} proposed by Kass and Madden in \cite{kass-madden}; we show that the method  of Kass and Madden is equivalent to a special case of the method of ordered epitomes.

\subsection{Plain epitomes}\label{s:plain-epitomes}
In this section we propose a general approach to the proof of non-soficness that captures the intuition behind the standard proof from Example~\ref{ex:mirror}. We begin with a definition of ``plain'' epitome functions.
\begin{definition}\label{d:epitome}
Let $\cal P$ be the set of all patterns on $\mathbb{Z}^2$ over some finite alphabet. Let
\[
{\cal E} \ : \ {\cal P} \mapsto \{0,1\}^*
\]
be a {partial} function mapping patterns to binary strings.  We say that $\cal E$ is an \emph{epitome function} for the shift $\cal S$ if for every pattern $P$ with support $B$ that is globally admissible in $\cal S$ (in brief, $P$ is a \emph{valid pattern}) such that ${\cal E}(P)$ is defined, there exists an \emph{epitome witness} pattern $R$ on $\mathbb{Z}^2 \setminus B$ (the complement of $B$) for which

\begin{itemize}
\item[(i)] $R$ is compatible with $P$, i.e., the union of $P$ and $R$ is a valid configuration in $\cal S$, and
\item[(ii)]  for every pattern $P'$ with support $B$ that is compatible with $R$ and for which ${\cal E}(P')$ is defined, ${\cal E}(P) = {\cal E}(P')$. In other words, the epitome witness $R$ uniquely determines the value of $\cal E$ for the valid patterns on the complement of its support.
\end{itemize}
Given an epitome function $\cal E$ and a pattern $P$, we refer to ${\cal E}(P)$ as the epitome of $P$, or as the $\cal E$-epitome of $P$ if the epitome function is not clear from the context.

In general, an epitome function can be defined for patterns of any shape.
However, in this paper we only use  epitome functions $\cal E$  defined on patterns with support $B_n$ (a square if size $n\times n$) for different $n>0$.
In certain situations it can be helpful to focus on the restriction of an epitome function onto the patterns of a fixed size.
Given an epitome function $\cal E$, we write ${\cal E}_n$ for the restriction of $\cal E$ onto the patterns with support $B_n$. We  refer to the sequence ${\cal E}_n$  ($n  = 1,2,\ldots$) as a \emph{family of restricted epitome functions for} $\cal S$.
If an epitome function is defined only for patterns with supports $B_n$ (which is always the case in this paper),
a family of restricted epitome functions ${\cal E}_n$ for all $n>0$ and an epitome function ${\cal E}$ are two equivalent languages to describe one and the same object.

We say that $\cal E$ is \emph{computable} if there is an algorithm that computes $\cal E$. As is customary in computability theory, we assume that an algorithm computing a partial function returns no result for inputs outside its domain. 
If, moreover,  there is an algorithm computing the value ${\cal E}(P)$ in time $2^{O(n^2)}$ for each input pattern $P$ of size $n\times n$ ($n=1,2,\ldots$) from the domain of $\cal E$, then we say that this  epitome function is exp-time computable.

Similarly, we say that an epitome function is \emph{computable with oracle $\cal O$} if 
the domain of  ${\cal E}$ is decidable with oracle $\cal O$ and
there is an algorithm  that computes  ${\cal E}$ given the access to oracle $\cal O$. We can also talk about  epitome functions that are \emph{exp-time computable with an oracle}.
\end{definition}

\begin{remark}
If an epitome function is exp-time computable, then the domain of ${\cal E}$ can be decided in exponential time:
it is enough to run the algorithm computing  ${\cal E}$  and abort it if the computation does not converge in due time. 
\end{remark}

\begin{remark}\label{r:abuse}
Formally speaking,  ${\cal E}$-epitomes are binary strings --- objects  for which we have defined Kolmogorov complexity. In  applications, it can be convenient to assign the values of an epitome function to integer numbers, to vectors of integer numbers, or to multi-dimensional finite patterns over a finite alphabet. All these constructive objects can be represented by finite binary strings (with a naturally chosen encoding); so in some cases we  abuse  notation and talk about epitomes that are  integer numbers (or vectors of integer numbers, or finite patterns, and so on) and  their Kolmogorov complexities.
\end{remark}

\begin{proposition}\label{p:epitomes} 
(a) For every sofic  shift with a computable epitome function ${\cal E}$,  for every globally admissible pattern $P$ of size $n\times n$ 
such that ${\cal E}(P)$ is defined, we have
 $
 \CK({\cal E}(P)) = O(n).
 $

(b) For every sofic shift with an exp-time computable epitome function ${\cal E}$, for every globally admissible pattern $P$ of size $n\times n$ such that ${\cal E}(P)$ is defined,  we have
 $
 \CK^{T(n)}({\cal E}(P)) = O(n)
 $
for a time threshold $T(n)=2^{O(n^2)}$.

(c) The statements (a) and (b) relativize. That is, for every sofic  shift with an epitome function ${\cal E}$ that is \emph{computable with oracle} $\cal O$ (or \emph{exp-time computable with oracle}  $\cal O$),  for every globally admissible pattern $P$ of size $n\times n$ such that ${\cal E}(P)$ is defined,  we have
 $
 \CK^{\cal O}({\cal E}(P)) = O(n)
 $
 (or, respectively,   $\CK^{T(n), \cal O}({\cal E}(P)) = O(n)$ for a time threshold $T(n)=2^{O(n^2)}$).

\end{proposition}

\begin{proof}
(a) Assume  that $\cal S$ is a sofic shift with a covering SFT $\hat {\cal S}$ (i.e., $\cal S$ is a coordinate-wise $\pi$-projection of $\hat {\cal S}$). 
Let $P$ be a pattern with support $B_n$  in $\cal S$ and $R$ be the pattern on the complement of $B_n$ that enforces the value 
of  the ${\cal E}$-epitome of $P$ (as specified in Definition~\ref{d:epitome}). Denote by $Y$ a configuration in $\hat {\cal S}$ whose $\pi$-projection gives the union of $P$ and $R$.
Let $Q$ be a pattern of size $n\times n$ in $Y$ such that $P$ is a coordinate-wise projection of $Q$.  
Denote by $\partial Q$ the ``border zone,'' i.e., the nearest neighborhood of $Q$ (the border zone of $Q$ is not included in $Q$; if $Q$ is an $n\times n$ square, then border zone $\partial Q$ around $Q$ consists of $4(n+1)$ letters, see Fig.~\ref{f:sft2sofic}).

We assume without loss of generality that the local constraints in $\hat {\cal S}$ involve only pairs of neighboring nodes in $\mathbb{Z}^2$. Then, every locally admissible pattern $Q'$ of size $n\times n$ that is locally compatible with border zone $\partial Q$, must be compatible with the rest of  configuration $Y$. Therefore, the $\pi$-projections of these $Q'$ are compatible with $R$. 
We  apply the definition of an epitome function: if  ${\cal E}$ is defined on  the $\pi$-projection of such $Q'$,  then  the epitome  of  $Q'$ must be equal to the epitome of  $P$.

\begin{figure}
\begin{tikzpicture}[scale=.6,on grid]

   \begin{scope}[
           yshift=-95,every node/.append style={
           yslant=0.5,xslant=-1},yslant=0.5,xslant=-1
           ]
       \draw[step=4mm, black] (0,0) grid (5,5);
       \draw[black,thick] (0,0) rectangle (5,5);
              \foreach \x in {-1,...,4} \foreach \y in {-1,...,4} 
           \fill[black!50] (1.65+0.4*\x,1.95+0.4*\y) rectangle (1.95+0.4*\x,1.65+0.4*\y); 
   \end{scope}

   \begin{scope}[
           xshift=0.0,yshift=0.5,every node/.append style={
           yslant=0.5,xslant=-1},yslant=0.5,xslant=-1
           ]
       \fill[white,fill opacity=0.80] (0,0) rectangle (5,5);
       \draw[step=4mm, black] (0,0) grid (5,5); 
       \draw[black,thick] (0,0) rectangle (5,5);

       \foreach \x in {-1,...,4} \foreach \y in {-1,...,4} 
           \fill[blue!70] (1.65+0.4*\x,1.95+0.4*\y) rectangle (1.95+0.4*\x,1.65+0.4*\y);

    \foreach \x in {-1,...,6} 
       \fill[red!41] (1.25+0.4*\x,1.55-0.4) rectangle (1.55+0.4*\x,1.25-0.4);
    \foreach \x in {-1,...,6} 
       \fill[red!41] (1.25+0.4*\x,1.55+6*0.4) rectangle (1.55+0.4*\x,1.25+6*0.4);
    \foreach \x in {-1,...,6} 
       \fill[red!41] (1.25-0.4+0*\x,1.55+\x*0.4) rectangle (1.55-0.4+0*\x,1.25+\x*0.4);
    \foreach \x in {-1,...,6} 
       \fill[red!41] (1.25+2+0.4,1.55+\x*0.4) rectangle (1.55+2+0.4,1.25+\x*0.4);
   \end{scope}

   \draw[-latex,ultra thick,black](-3,5)node[left]{$n\times n$ pattern $Q$}
       to[out=0,in=90] (0,2.6);
   \draw[-latex,ultra thick,black](-4,4)node[left]{border zone  $\partial Q$ }
       to[out=0,in=90] (-.8,1.4);

   \draw[thick,black!30,dashed] (2.0,2.4)  -- (2.0,1.0);
   \draw[-latex,thick,black,dashed](2.0,1.0)  to (2.0,-1.0);
   \node[color=black] at (2.4,0.85) {$\pi$};
   \node [draw,circle,inner sep=0.5pt,fill] at (2.0,2.4) {};

   \draw[-latex,ultra thick,black](-4,-3)node[left]{$n\times n$ pattern $P$}
       to[out=0,in=200] (-0.5,-1.3);
   \draw[thick,black](6,4.5) node {a configuration in an SFT};
   \draw[thick,black](6.3,-2.8) node {a configuration in a sofic shift};

\end{tikzpicture}
\caption{Projection of an $n\times n$ pattern from an SFT onto a sofic shift.} \label{f:sft2sofic}
\end{figure}

It follows that ${\cal E}(P)$ can be computed  given only the coloring of   $\partial Q$:  we use the brute-force search to find  all patterns  $Q'$ that are locally compatible with $\partial Q$, apply to them projection $\pi$, and run in parallel the computation of  their epitomes. As soon as we find at least one $Q'$ compatible with  $\partial Q$ such that the epitome function is defined on its $\pi$-projection, we are done:
though the  projection $\pi(Q')$ can be different from $P$, the epitome of $\pi(Q')$  must coincide with the epitome of $P$. Since the size of $\partial Q$ is linear in $n$, we conclude that 
$
\CK({\cal E}(P)) = O(n).
$

To prove claim~(b), we observe that the brute force search over all patterns of size $n\times n$ runs in time $2^{O(n^2)}$, so we get
$
\CK^{2^{O(n^2)}}({\cal E}(P)) = O(n).
$
It is easy to verify that the proof relativizes (each step of the argument works with any fixed oracle $\cal O$), so we get claim~(c).
\end{proof}

Proposition~\ref{p:epitomes} gives a necessary condition for  soficness. To prove that a shift is not sofic, we need to provide a computable
(or an exp-time computable) epitome function with high Kolmogorov complexities (or, respectively, resource-bounded Kolmogorov complexities).
In what follows we discuss simple examples of application of this technique. We start with the standard example of mirror symmetric shift.

\begin{examplerev}{ex:mirror}
\label{ex:mirror-revisited}
 Let ${\cal S}_{\text{mirror}}$ be the shift from Example~\ref{ex:mirror} in the introduction (the mirror-symmetric configurations).
For this example we can define  an epitome function ${\cal E}$ as follows:
\begin{itemize}
\item if an $n\times n$ pattern $P$ contains only black and white letters, then ${\cal E}(P)$ is a binary string of length $n^2$
that identifies $P$ uniquely (roughly speaking, ${\cal E}$ does not compress the patterns in black and white);
\item  ${\cal E}$ is undefined for all patterns containing at least one red letter.
\end{itemize}
The defined  ${\cal E}$ satisfies the definition of a computable epitome function, since the part of a configuration above the red line determines uniquely all black-and-white patterns below this line. 
More specifically, for every $n\times n$ pattern $P$  that contains only black and white letters,  
let $R$ be the infinite pattern defined as follows:
\begin{itemize}
\item the domain of $R$ is the complement of the $n\times n$ domain of $P$,
\item the horizontal line in $R$ that is one row above $P$ must consist of only red letters,
\item the positions that are symmetric (with respect to the red line) to the positions with  black letters in $P$ must be  black,
\item all other letters in $R$ must be white,
\end{itemize}
as shown in Fig.~\ref{f:p-and-r}. It is clear that $R$ is compatible with $P$ and not compatible with other $n\times n$ patterns,
as required in the definition of an epitome function.

Since there are $2^{n^2}$ patterns of size $n\times n$ with black and white letters, for some patterns of size $n\times n$ we have
 $
 \CK({\cal E}(P)) \ge n^2.
 $
Therefore, we can apply Proposition~\ref{p:epitomes}~(a) and conclude that the shift is not sofic.

\begin{figure}[H]

  \centering
  \begin{tikzpicture}[scale=0.25,x=1cm,baseline=2.125cm]

  \pgfmathsetseed{1}
    \foreach \x in {1,...,12} \foreach \y in {1,...,7}
    {
        \pgfmathparse{mod(int(random*23),2) ? "black!10" : "black!66"}
        \edef\colour{\pgfmathresult}
        \path[fill=black!10,draw=black] (\x,7-\y) rectangle ++ (1,1);
        \path[fill=black!10,draw=black] (\x,7+\y) rectangle ++ (1,1);
    }

    \foreach \x in {6,...,10} \foreach \y in {1,...,5}
    {
        \pgfmathparse{mod(int(random*23),2) ? "black!10" : "black!66"}
        \edef\colour{\pgfmathresult}
        \path[fill=\colour,draw=black] (\x-20,7-\y) rectangle ++ (1,1);
        \path[fill=\colour,draw=black] (\x,7-\y) rectangle ++ (1,1);
        \path[fill=\colour,draw=black] (\x,7+\y) rectangle ++ (1,1);
    }

     \foreach \x in {1,...,12} 
    {
        \path[fill=red!41,draw=black] (\x,7) rectangle ++ (1,1);
    }

\path[fill=white,draw=blue, ultra thick] (6,5-3) rectangle ++ (5,5);
\path[draw=orange, ultra thick, dashed] (6,4-3+7+0.1) rectangle ++ (5,5);
\path[draw=blue, ultra thick] (6-20,5-3) rectangle ++ (5,5);

    \draw[-latex,ultra thick,blue!95](-18,10.5)node[left]{$P$}    to[out=0,in=90] (-11.5,7.3);    
    \draw[-latex,ultra thick,blue!95](20,3.5)node[below]{$R$}    to[out=90,in=0] (13.2,7.5); 

    \draw[-latex,ultra thick,blue!95, dashed](-11.5,3.8)node[below]{}    to[out=-30,in=-135] (8.0,4.2);

\end{tikzpicture}

\caption{An  $n\times n$ pattern $P$ with black and white letters and the corresponding pattern $R$ with the complementary domain.
The orange (dashed) frame contains the pattern that is mirror symmetric to $P$.} 
\label{f:p-and-r}

\end{figure}
\end{examplerev}

\begin{example}[Mirror shift with  low plain Kolmogorov complexity]
\label{ex:mirror-low-complexity}
 Let us consider a subshift of ${\cal S}_{\text{mirror}}$ from Example~\ref{ex:mirror}:
we still admit only symmetric configurations, but we now allow only those $n\times n$ patterns $P$ 
in black and white that are globally admissible for the shift ${\cal S}_{0}$  introduced in Notation~\ref{d:s_0},
p.~\pageref{d:s_0}.
 A non-degenerate configuration of the new shift looks as follows:
there is an infinite horizontal line composed of red letters, and  two symmetric half-planes (filled with black and white letters) above and below this line,  with $n\times n$ patterns $P$ such that
\[
 \CK^{\hat T(n)}(P) \le \hat \lambda \log n
\]
(the choice of $\hat \lambda$ and $ \hat T(n)$ is explained before Notation~\ref{d:s_0}).
 The new shift is effective, and the number of globally admissible patterns  is at most
$
2^{O(\log n)}=\poly(n).
$
We know also  that \emph{some} globally admissible $n\times n$ patterns in this shift satisfy
\begin{equation}\label{eq:super-linear-lower-bound}
\CK^{2^{n^3}}(P) =\Omega(n^{1.5}). 
\end{equation}

We cannot apply Theorem~\ref{th:dls}  directly and conclude that the new shift is non-sofic. 
The obstacle  is that this shift  admits some patterns with very
low time-bounded complexity; for example, the shift admits the configuration with an infinite horizontal line in red and only white letters above and below this line.  Proposition~\ref{p:epitomes}~(a) also does not apply since the plain Kolmogorov complexity of all admissible patterns (and, therefore, of the epitomes of these patterns) is logarithmic. However, in this case we can use exp-time computable epitome functions.

The functions ${\cal E}$ defined for Example~\ref{ex:mirror} (see p.~\pageref{ex:mirror-revisited}) provide for this shift an \emph{exp-time computable}  epitome function. 
For some (though not for all) $n\times n$ patterns $P$ we have \eqref{eq:super-linear-lower-bound},
so it follows from Proposition~\ref{p:epitomes}~(b) that the shift is not sofic.
\end{example}

\begin{remark}
In the definition of epitome functions we permit that ${\cal E}$ can be defined for non-admissible patterns. 
In the next example we show that this possibility  can be useful.
\end{remark}
\begin{figure}
\centering

  \begin{tikzpicture}[scale=0.27,x=1cm,baseline=2.125cm]
  
  \pgfmathsetseed{5}
    \foreach \x in {1,...,8} \foreach \y in {1,...,8}
    {
        \pgfmathparse{mod(int(random*23),2) ? "black!10" : "black!66"}
        \edef\colour{\pgfmathresult}
        \path[fill=\colour,draw=black] (\x,\y) rectangle ++ (1,1);
    }
    
     \foreach \x in {1,...,8}
    {
    \path[fill=red!41,draw=black] (\x,0) rectangle ++ (1,1);
    \path[fill=red!41,draw=black] (\x,9) rectangle ++ (1,1);    
    }    

     \foreach \y in {0,...,9}
    {
    \path[fill=red!41,draw=black] (0,\y) rectangle ++ (1,1);
    \path[fill=red!41,draw=black] (9,\y) rectangle ++ (1,1);    
    }

\end{tikzpicture}
  \caption{An example of a special pattern of size $n\times n$ from Example~\ref{ex:busy-beaver}:  the border line (included in the pattern) consists of red letters, and the remaining $(n-2)\times (n-2)$ letters can be only black or white.}
  \label{fig:special}
\end{figure}  
\begin{example}
\label{ex:busy-beaver}
In this example we  use again the alphabet of three letters (\emph{red}, \emph{black}, and \emph{white}). Before we define a shift on this alphabet, we need to introduce some terms.
We say that a pattern  $P$ of size $n\times n$ is \emph{special} if it contains only red letters on the border line,  and only black and white letters in the internal points, as shown in Fig.~\ref{fig:special}. It is convenient to require that a special pattern contains at least one black letter, so there are $2^{(n-2)^2}-1$ special patterns of size $n\times n$.
For every $n$ we fix a one-to-one correspondence between the special patterns of this size and the binary strings of length smaller than $(n-2)^2$. We denote  by $x_P$ the binary string corresponding to pattern $P$.

Let $U$ be the computable function from the definition of Kolmogorov complexity (see Theorem~\ref{t:kolmogorov-invariance}). Fix an algorithm computing $U$. It is known that  function $U$ is partial, so the algorithm computing $U$ diverges on some inputs. 
For every integer number $m> 0$ we denote by $\hat p_m$ the string of length smaller than $m$ such that  the computation of $U(\hat p_m)$ converges but takes longer than the computation of $U(p)$ for any other $p$ with $|p|<m$. In case of a tie (if there are several $p$ requiring the same number of steps to complete the computation) we chose $\hat p_m$ that is lexicographically biggest.

Now we define a shift on $\mathbb{Z}^2$  by the following rule: for every $n$, a special pattern $P$ of size $n\times n$ is globally admissible if   $U(x_P)$ is not defined  or if $x_P = \hat p_{(n-2)^2}$ (i.e., if the pattern corresponds to the longest converging computation of $U$ among all inputs of length less than $(n-2)^2$). There is no other constraints:  a configuration belongs to the shift if and only if it does not contain any  special pattern forbidden by this rule.

This shift is effective:  the set of inputs $p$ where the computation of $U(p)$ converges is recursively enumerable and, therefore, so is  the  set of forbidden special patterns. Now we prove that this shift is not sofic.
To this end, we define an epitome function ${\cal E}$ as follows: 
\begin{itemize}
\item for every special pattern $P$ of size $n$ we let  ${\cal E}(P)= P$ if  $U(x_P)$  is defined;
\item  ${\cal E}$ is not defined for other patterns.
\end{itemize}
This mapping ${\cal E}$ is obviously computable.
Observe that  ${\cal E}$ is defined for many inadmissible patterns and  for one and only one  globally admissible pattern of size $n\times n$ (namely, for the pattern representing the longest-running computation of $U$).  In this situation, the definition of an epitome function is trivially satisfied.

Given $\hat p_{(n-2)^2}$, we can find the maximal number of steps required for a computation of $U(p)$ to converge  for $|p| <(n-2)^2$.
This information permits to  find all converging computations of $U(p)$ for $|p| <(n-2)^2$,  to reveal all strings $x$ with Kolmogorov complexity strictly smaller than $(n-2)^2$, and then to take the first incompressible string of length $(n-2)^2$. 
In other words, given $\hat p_{(n-2)^2}$ we can   find an object with Kolmogorov complexity at least $(n-2)^2$. Hence, 
Kolmogorov complexity of $\hat p_{(n-2)^2}$ itself is  close to $(n-2)^2$.
Thus, for the only $n\times n$ pattern $P$ that is globally admissible and where  ${\cal E}(P)$ is defined, we have
\[
\CK({\cal E}(P)) = \Omega(n^2).
\]
It follows from Proposition~\ref{p:epitomes}~(a) that the shift is not sofic. 

Let us note in conclusion that we used in this example epitome functions defined on many non-admissible patterns, and that the domain of ${\cal E}$  is recursively enumerable but not decidable. These features seem to be essential for this argument.
\end{example}

\subsection{Ordered epitomes}\label{s:ordered-epitomes}
The argument based on Definition~\ref{d:epitome} does not apply to \cite[Example~2.5]{kass-madden} and similar shifts. To handle this class of non-sofic shifts, we introduce a slightly more general definition of epitome functions:
\begin{definition}
\label{d:ordered-epitome}
Let $\cal P$ be the set of all patterns on $\mathbb{Z}^2$ over some finite alphabet, 
 $\preccurlyeq$ be a partial order on $\{0,1\}^*$,  and
\[
{\cal E} \ : \ {\cal P} \mapsto \{0,1\}^*
\]
be a \emph{partial} function. We say that ${\cal E}$ (together with the  assigned partial order $\preccurlyeq$)  is an \emph{o-epitome function} (a function computing \emph{ordered} epitomes) for a shift $\cal S$, 
if for every globally admissible pattern $P$ with support $B$  such that  ${\cal E}(P)$ is defined, there exists a pattern $R$ on $F = \mathbb{Z}^2\setminus B$
such that
\begin{itemize}
\item[(i)] $R$ is compatible with $P$, i.e., the union of $P$ and $R$ provides a valid configuration in $\cal S$, and
\item[(ii)]  for every pattern $P'$  on support $B$ compatible with $R$, if ${\cal E}(P')$ is defined  then
\[
{\cal E}(P')  \preccurlyeq  {\cal E}(P) 
\]
\end{itemize}
\textup(i.e., this configuration $R$ on the complement of $B$ determines the  maximum of the ${\cal E}$-epitomes over all valid $P'$\textup).

To denote  an o-epitome function (a mapping considered together with a partial order) we write $({\cal E}, \preccurlyeq)$. 
We say that values of $\cal E$  are \emph{ordered  epitomes} or \emph{o-epitomes}.

Similarly to the case of plain epitomes,  in general, ordered epitomes can be defined for patterns of any shape. However, in this paper we use only o-epitome functions defined on patterns with support $B_n$ (square of size $n\times n$) for $n>0$. It can be helpful to consider the restriction of an o-epitome function onto the patterns of a fixed size.
Given an o-epitome function $(\cal E, \preccurlyeq)$, we write ${\cal E}_n$ and and $ \preccurlyeq_n$ for the restrictions of $\cal E$ and $ \preccurlyeq$ respectively onto the patterns with support $B_n$. We  refer to the sequence $({\cal E}_n,\preccurlyeq_n)$  ($n  = 1,2,\ldots$) as a \emph{family of restricted o-epitome functions}\footnote{%
If an  o-epitome function $({\cal E},  \preccurlyeq)$ is defined only for patterns with supports $B_n$, then it is determined uniquely by the corresponding family of restricted o-epitome functions $({\cal E}_n,\preccurlyeq_n)$.}.

We say that an o-epitome function $(\cal E,  \preccurlyeq)$  is \emph{computable} 
if there are algorithms  that compute the mapping ${\cal E}$ and the relations $\preccurlyeq$. 
If, moreover,  ${\cal E}$ and $\preccurlyeq$ are computable in time $2^{O(n^2)}$ (for patterns of size $n\times n$), we say that this  o-epitome function is \emph{exp-time computable}. 
As usual, an algorithm computing  a partial function must not give any result for inputs outside its domain.

Similarly, we say that an o-epitome function is \emph{computable with oracle $\cal O$} if there are algorithms  that 
can compute the mappings ${\cal E}$ and the relations $\preccurlyeq$ given the access to oracle $\cal O$. 
We can also talk about o-epitome functions that are \emph{exp-time computable} with an oracle.
\end{definition}

Plain epitome functions from Definition~\ref{d:epitome}  can be viewed as a special case of Definition~\ref{d:ordered-epitome}.
If ${\cal E}$ is a computable (or exp-time computable) epitome in the sense of  Definition~\ref{d:epitome}
and if $\preccurlyeq$ is an arbitrary effectively  computable (exp-time computable) order on the ${\cal E}$-epitomes, then $({\cal E},\preccurlyeq)$ is a computable (or, respectively, an exp-time computable)  o-epitome function in the sense of Definition~\ref{d:ordered-epitome}.
In Definition~\ref{d:epitome} the neighborhood $R$ enforces the exact value of  ${\cal E}(P')$  over all  $P'$ compatible with $R$, 
while in  Definition~\ref{d:ordered-epitome} $R$ enforces only the maximum of ${\cal E}(P')$.

Remark~\ref{r:abuse} applies to the o-epitome functions: we  may abuse  notation and talk about ordered epitomes that are  integer numbers (or vectors of integer numbers, or finite patterns, and so on) and  their Kolmogorov complexities.

\begin{proposition}\label{pbis:epitomes} 
(a) For every sofic  shift with a computable o-epitome function $({\cal E},\preccurlyeq)$,  for every globally admissible pattern $P$ of size $n\times n$ such that ${\cal E}(P)$ is defined, we have
 $
 \CK({\cal E}(P)) = O(n).
 $
 
(b) For every sofic shift with an exp-time computable  o-epitome function $({\cal E},\preccurlyeq)$, for every globally admissible pattern $P$ of size $n\times n$ such that ${\cal E}(P)$ is defined, we have
 $
 \CK^{T(n)}({\cal E}(P)) = O(n)
 $
for a time threshold $T(n)=2^{O(n^2)}$.

(c) The statements (a) and (b) relativize. That is, for every sofic  shift with an o-epitome function $({\cal E},\preccurlyeq)$ that is \emph{computable with oracle} $\cal O$ (or \emph{exp-time computable with oracle}  $\cal O$),  for every globally admissible pattern $P$ of size $n\times n$ such that ${\cal E}(P)$ is defined,  we have
 $
 \CK^{\cal O}({\cal E}(P)) = O(n)
 $
 (or, respectively,   $\CK^{T(n), \cal O}({\cal E}(P)) = O(n)$ for a time threshold $T(n)=2^{O(n^2)}$).

\end{proposition}
\begin{proof}
We reuse  the notation from the proof of Proposition~\ref{p:epitomes}: $\cal S$ denotes a sofic subshift, $\hat {\cal S}$ denotes an SFT covering $\cal S$, and $\pi$ is a letter-to-letter projection mapping $\hat {\cal S}$ onto $\cal S$. We may assume that the local constraints in $\hat {\cal S}$ involve only pairs of neighboring nodes in $\mathbb{Z}^2$.
For an $n\times n$ pattern $P$ in $\cal S$, we denote by $Q$ the $n\times n$ pattern in $\hat {\cal S}$ projected on $P$, and $\partial Q$ denotes the \emph{border zone} (nearest neighborhood) around $Q$ consisting of $4(n+1)$ letters, as shown in Fig.~\ref{f:sft2sofic}. We call such patterns $\partial Q$   $n$-\emph{border zones}.

\smallskip

We start with a proof of (b), which is very similar to the proof of Proposition~\ref{p:epitomes}. In the previous proof, it was enough to find \emph{at least one} pattern $Q'$ compatible with the given border zone $\partial Q$ such that the epitome of $\pi(Q')$ is defined; then computing the epitome of $\pi(Q')$ we obtain the epitome of $P$. Now we should find \emph{all} patterns $Q'$ compatible with $\partial Q$, apply to each of them the projection $\pi$, compute their epitomes (for those patterns where ${\cal E}$ is defined), and then take the maximum of the obtained results. 
The fact that ${\cal E}$ are partial functions provides no obstacle. Indeed, we assumed that the o-epitome function is computable in exponential time; so we can safely abort the computations  that do not converge in due time.
It only remains to observe that the exhaustive search over all patterns of size $n\times n$  requires to examine $2^{O(n^2)}$ possibilities, which can be performed in exponential time.

\smallskip

We cannot prove~(a) in the same way because  the computation of epitomes can take arbitrary long time, and we cannot find algorithmically the finalized list of ${\cal E}(P')$ for all $P'$ compatible with $\partial Q$.  In this case we need a more involved argument.

\begin{claim}\label{claim:1}
The number of $n$-border zones $\partial Q$ that are globally admissible in $\hat {\cal S}$ is not greater than $2^{O(n)}$.
\end{claim}
\begin{proof}
This claim is pretty trivial: the domain of such a border zone consists of $4(n+1)$ nodes, and the number of all  patterns
of this size (admissible or not) is at most $|\Sigma|^{4(n+1)}$, where $\Sigma$ is the alphabet of $\hat {\cal S}$.
\end{proof}

\begin{claim}\label{claim:2}
The number of values ${\cal E}(P)$ for globally admissible in $\cal S$ patterns $P $ is not bigger than the number 
of all  $n$-border zones $\partial Q$  that are globally admissible in $\hat {\cal S}$.
\end{claim}
\begin{proof}
By the definition of an o-epitome function, for every value ${\cal E}(P)$ there is a pattern $R$ on $F_n$ that determines implicitly   ${\cal E}(P)$: the maximum of ${\cal E}(P')$ for all $P'$ compatible with $R$ is equal to  ${\cal E}(P)$. Though $R$ is infinite, the ``interaction'' between $P$ and $R$ goes through a contour of linear length. 
As we observed in the proof of Proposition~\ref{p:epitomes}, we can choose an $n$-border zone $\partial Q$ (globally admissible in $\hat {\cal S}$) 
with the following property:
if we take all $n\times n$ patterns $Q$ locally compatible with $\partial Q$ in $\hat {\cal S}$, and then take $\pi$-projections of these patterns, we obtain the required pattern $P$ and possibly several other 
(globally admissible in $\cal S$) patterns $P'$ such that ${\cal E}(P')$,  if defined,  is inferior or equal to ${\cal E}(P)$ in the sense of $\preccurlyeq$.
Therefore, the chosen $n$-border zone $\partial Q$ determines uniquely the value of ${\cal E}(P)$. 

Observe that  $\partial Q$ determines  ${\cal E}(P)$ in an abstract sense; 
we cannot compute  ${\cal E}(P)$ from $\partial Q$ algorithmically. Indeed,   since ${\cal E}$ is a partial function (and the algorithm computing ${\cal E}$  does not stop on inputs where the function is not defined),  we cannot find algorithmically the maximum value of ${\cal E}(P')$ for all $P'$ obtained from $\partial Q$. However, we can claim that the number of values ${\cal E}(P)$ for globally admissible (in shift $\cal S$) patterns $P$ is not greater than the number of globally admissible  (in shift $\hat {\cal S}$) $\partial Q$.
\end{proof}

The set of  $n$-border zones  that are globally admissible  in any SFT (actually, in any effective shift) is co-enumerable (the complement is recursively enumerable): we can examine progressively all patterns of bigger and bigger size, detect those  that are forbidden, and reveal one by one the border zones that are inadmissible (because all their extensions contain forbidden patterns).  
Using this procedure we can eventually reveal each  non-admissible pattern;
but \emph{a priori} we do not know when the  last inadmissible $n$-border zone is detected.

Denote by $N_n$ the number of $n$-border zones that are globally admissible in $\hat {\cal S}$. Given this number (its binary expansion) we can determine the moment when the described algorithmic process finalizes the list of globally admissible  $n$-border zones  (i.e., when all inadmissible $n$-border zones are revealed). Then, given the list of all non-eliminated (globally admissible) $n$-border zones, we can find all $n\times n$ patterns in $\hat S $ that are locally compatible with at least one globally admissible  $n$-border zone (so these patterns are also globally admissible). Then we apply to each of these patterns the mapping $\pi$ and obtain the $n\times n$ patterns that are globally admissible in $\cal S$.  Finally, we apply the algorithm computing ${\cal E}$ to each  obtained pattern and run these computations in parallel. Every value of ${\cal E}(P)$ for globally admissible $P$ will be eventually found by this procedure.

Thus, to describe one specific value of ${\cal E}$, we need to describe the  enumeration procedure explained above and  to specify the ordinal number of the sought value of ${\cal E}$ in this enumeration (in the order of appearance). To run the enumeration process, we need to know the numbers $n$  and $N_n$. Due to Claim~\ref{claim:1}, the binary expansion of $N_n$ consists of only $O(\log N_n) = O(n)$ bits.
The number of different values of ${\cal E}$ for globally admissible patterns is bounded by $N_n$ (Claim~\ref{claim:2}); so the ordinal number of an element in this enumeration can be specified by $\log N_n = O(n)$ bits (again, due to Claim~\ref{claim:1}). Thus, every value of  ${\cal E}$ can be described by an algorithm of size $O(n)$, which is the claim~(a) of the theorem.

\smallskip

The arguments used in the proofs of (a) and (b) obviously relativizes, and we get statement~(c) for computations with an oracle.
\end{proof}

Proposition~\ref{pbis:epitomes} provided us with another necessary condition for soficness. To prove that a shift is not sofic, it is enough to construct a computable (or exp-time computable) o-epitome function with super-linear resource-bounded Kolmogorov complexities.
\begin{remark}
When we provide an algorithm with an access to an oracle, its computational power can increase dramatically. Thus, when we say that 
 $\CK^{ \cal O}(w) = O(n)$ with some (unspecified) oracle $\cal O$, we do not say much about string $w$. However, for every oracle we still can say that there are it most $2^{O(n)}$ objects of complexity $O(n)$. So when we claim that $\CK^{ \cal O}({\cal E}(P)) = O(n)$  for all $P$, we say essentially that there are $2^{O(n)}$  values of epitomes ${\cal E}$. In other words, the bounds on Kolmogorov complexity with an oracle (and without a time bound) is only a different language for the usual counting argument.
On the other hand, the meaning of  Kolmogorov complexity  becomes much subtler  when we constrain  the computational resources of  algorithms in use; the claims $\CK^{T(n)}(w) = O(n)$  or $\CK^{T(n),{\cal O}}(w) = O(n)$ cannot be reduced to a simple counting.
\end{remark}

\begin{example}[semi-mirror shift]
\label{ex:semi-mirror}
In this example we define a shift space extending the one defined in Example~\ref{ex:mirror}.
Let $\Sigma$ be the alphabet with three letters (e.g., \emph{black}, \emph{white}, and \emph{red});
the admissible configurations of the shift are all black-and-white configurations (without any red letter) and 
the configurations with an infinite horizontal line of red letters and  black-and-white half-planes 
above and below this line that are \emph{semi-symmetric} in the following sense:
if the letter in the $i$th column at the distance $k$ \emph{below} the horizontal red line is black, then 
the letter  in the same $i$th column at the distance $k$ \emph{above} the red line must be also black.
In other words, we may take symmetric half-planes above and below the red line, and then convert some
of the black letters in the lower half-plane into white ones,  see Fig.~\ref{f:semi-mirror}. 
We denote this shift by  ${\cal S}_{\text{semi-mirror}}$.

\begin{figure}[H]
\begin{center}

  \centering
  \begin{tikzpicture}[scale=0.25,x=1cm,baseline=2.125cm]

  \pgfmathsetseed{2}
    \foreach \x in {1,...,12} \foreach \y in {1,...,7}
    {
        \pgfmathparse{mod(int(random*23),2) ? "black!10" : "black!66"}
        \edef\colour{\pgfmathresult}
        \path[fill=\colour,draw=black] (\x,7-\y) rectangle ++ (1,1);
        \path[fill=\colour,draw=black] (\x,7+\y) rectangle ++ (1,1);
    }
    
     \foreach \x in {1,...,12} 
    {
        \path[fill=red!41,draw=black] (\x,7) rectangle ++ (1,1);
    }

        \path[fill=black!10,draw=blue,ultra thick] (1,7-4) rectangle ++ (1,1);
        \path[fill=black!10,draw=blue,ultra thick] (3,7-3) rectangle ++ (1,1);
        \path[fill=black!10,draw=blue,ultra thick] (8,7-2) rectangle ++ (1,1);
        \path[fill=black!10,draw=blue,ultra thick] (9,7-6) rectangle ++ (1,1);

        \path[fill=black!66,draw=blue,ultra thick] (1,7+4) rectangle ++ (1,1);
        \path[fill=black!66,draw=blue,ultra thick] (3,7+3) rectangle ++ (1,1);
        \path[fill=black!66,draw=blue,ultra thick] (8,7+2) rectangle ++ (1,1);
        \path[fill=black!66,draw=blue,ultra thick] (9,7+6) rectangle ++ (1,1);

\end{tikzpicture}

\caption{A configuration of the shift ${\cal S}_{\text{semi-mirror}}$ (we show in blue the positions where the upper and lower half-planes mismatch).} \label{f:semi-mirror}

\end{center}
\end{figure}

It is easy to see  that this shift is effective: the forbidden patterns are those where the red letters are not horizontally aligned, 
and those where the areas above and below the horizontal red line do not satisfy the property of semi-symmetry.
Below  we prove that this shift is not sofic. To this end we use the technique of ordered epitomes.

We define  the epitome function  ${\cal E}$ trivially for the patterns of size $n\times n$ that contain only black and white letters.
For these patterns, we let  ${\cal E}(P)$  be $P$ (or its binary encoding). 
For  patterns with at least one red letter,  ${\cal E}$ is not defined. With this definition, for each $n>0$
there are $2^{n^2}$ possible values of ${\cal E}(P)$ for $n\times n$ patterns.

We define the partial order $\preccurlyeq$ as follows: ${\cal E}(P_1) \preccurlyeq {\cal E}(P_2)$ if $P_1$ and $P_2$ are patterns of the same size $n\times n$, and  for every position with a black letter in  $P_1$
the corresponding position in $P_2$ also contains a black letter, 
see  Fig.~\ref{f:p1-le-p2}.

\begin{figure}[H]
\vspace{5pt}
  \centering
  \begin{tikzpicture}[scale=0.27,x=1cm,baseline=2.125cm]
  
  \pgfmathsetseed{1}
    \foreach \x in {1,...,6} \foreach \y in {1,...,6}
    {
        \pgfmathparse{mod(int(random*23),2) ? "black!10" : "black!66"}
        \edef\colour{\pgfmathresult}
        \path[fill=\colour,draw=black] (\x,7-\y) rectangle ++ (1,1);
        \path[fill=\colour,draw=black] (\x+15,7-\y) rectangle ++ (1,1);
    }

        \path[fill=black!10,draw=black] (3,5) rectangle ++ (1,1);
        \path[fill=black!66,draw=black] (3+15,5) rectangle ++ (1,1);

        \path[fill=black!10,draw=black] (4,6) rectangle ++ (1,1);
        \path[fill=black!66,draw=black] (4+15,6) rectangle ++ (1,1);

        \path[fill=black!10,draw=black] (1,4) rectangle ++ (1,1);
        \path[fill=black!66,draw=black] (1+15,4) rectangle ++ (1,1);

        \path[fill=black!10,draw=black] (6,2) rectangle ++ (1,1);
        \path[fill=black!66,draw=black] (6+15,2) rectangle ++ (1,1);

           \draw[thick,black](12,3) node {\Large $\preccurlyeq$};

           \draw[thick,black](4.5,-0.5) node {\Large $P_1$};
           \draw[thick,black](19.5,-0.5) node {\Large $P_2$};

\end{tikzpicture}

\caption{A pair of patterns $P_1$ and $P_2$: the relation ${\cal E}(P_1)\preccurlyeq {\cal E}(P_2)$ means that the set of black position in $P_1$ is included in the set of black positions in $P_2$.} \label{f:p1-le-p2}
\vspace{-5pt}

\end{figure}

Let us verify that $({\cal E}, \preccurlyeq)$ satisfies the definition of an o-epitome function. For every $n\times n$ pattern $P$ with black and white letters  we define the corresponding pattern $R$ in exactly the same way as in Example~\ref{ex:mirror} on p.~\pageref{ex:mirror-revisited}, see Fig.~\ref{f:p-and-r}.
It is clear that when we join $P$ with the corresponding $R$, we obtain an admissible configuration of  the shift.  Further,  if we join the same $R$  with some other  pattern $P'$ (as in Fig.~\ref{f:p-and-p'}), we obtain an admissible configuration if and only if $P'$ together with $R$ form a semi-symmetric configuration, i.e., a letter in $P'$ can be black \emph{only if} the letter at the symmetric position in $R$ is black. Hence, $R$ is compatible with $P'$ if and only if  ${\cal E}(P') \preccurlyeq {\cal E}(P)$, which is the property required in the definition of an o-epitome function.

\begin{figure}

  \centering
  \begin{tikzpicture}[scale=0.27,x=1cm,baseline=2.125cm]
  
 \pgfmathsetseed{2}

\foreach \x in {-3,...,10} \foreach \y in {-9,...,9}
{
  \path[fill=black!10,draw=black] (\x,7-\y) rectangle ++ (1,1);
  \path[fill=black!10,draw=black] (\x+25,7-\y) rectangle ++ (1,1);
}
     \foreach \x in {-3,...,10} 
    {
        \path[fill=red!41,draw=black] (\x,7) rectangle ++ (1,1);
        \path[fill=red!41,draw=black] (\x+25,7) rectangle ++ (1,1);
    }

    \foreach \x in {1,...,6} \foreach \y in {1,...,6}
    {
        \pgfmathparse{mod(int(random*23),2) ? "black!10" : "black!66"}
        \edef\colour{\pgfmathresult}
        \path[fill=\colour,draw=black] (\x,7-\y) rectangle ++ (1,1);
        \path[fill=\colour,draw=black] (\x,7+\y) rectangle ++ (1,1);

        \path[fill=\colour,draw=black] (\x+25,7-\y) rectangle ++ (1,1);
        \path[fill=\colour,draw=black] (\x+25,7+\y) rectangle ++ (1,1);

    }

 \path[fill=black!10,draw=green,ultra thick] (3,7-4) rectangle ++ (1,1);
 \path[fill=black!10,draw=green,ultra thick] (5,7-2) rectangle ++ (1,1);
 \path[fill=black!66,draw=green,ultra thick] (3,7+4) rectangle ++ (1,1);
 \path[fill=black!66,draw=green,ultra thick] (5,7+2) rectangle ++ (1,1);

 \path[fill=black!66,draw=green,ultra thick] (3+25,7-4) rectangle ++ (1,1);
 \path[fill=black!66,draw=green,ultra thick] (5+25,7-2) rectangle ++ (1,1);
 \path[fill=black!66,draw=green,ultra thick] (3+25,7+4) rectangle ++ (1,1);
 \path[fill=black!66,draw=green,ultra thick] (5+25,7+2) rectangle ++ (1,1);

 \path[draw=blue,ultra thick] (1,1) rectangle ++ (6,6); %
 \path[draw=blue,ultra thick] (1+25,1) rectangle ++ (6,6);

                                                      
    \draw[-latex,ultra thick,orange!95](15,19)node[right]{$R$}    to[out=170,in=75] (4.5,15.0);                                               
    \draw[-latex,ultra thick,orange!95](17,19)node[left]{}    to[out=10,in=+105] (27.5,15.0);     

    \draw[-latex,ultra thick,blue!95](1,-4.5)node[left]{$P'$}    to[out=0,in=-90] (3.5,1.0);    
    \draw[-latex,ultra thick,blue!95](32,-4.5)node[right]{$P$}    to[out=180,in=-90] (29.5,1.0);

\end{tikzpicture}
\caption{On the right: a black-and-white pattern $P$ (in the blue frame) joined with the corresponding pattern $R$ with the complementary domain (outside of the blue frame). On the left: another black-and-white pattern $P'$ (such that $P'\preccurlyeq P$) joined with the same $R$.  We show in green the positions that distinguish $P$ and $P'$, and the corresponding positions in $R$ (symmetric with respect to the red line).} \label{f:p-and-p'}

\end{figure}

We have observed that for every $n$ we have $2^{n^2}$ possible values of ${\cal E}$. From the counting argument it follows that for every $n$ there exists at least one  pattern $P$ of size $n\times n$ such that $\CK({\cal E}(P)) \ge n^2 \gg n$. Thus, the property $\CK({\cal E}(P)) =O(n)$ cannot be true for this shift. From Proposition~\ref{pbis:epitomes}~(a) it follows that the shift is not sofic.
\end{example}

A non-degenerate configuration of shift  ${\cal S}_{\text{semi-mirror}}$ defined in Example~\ref{ex:semi-mirror}
 consists of two ``semi-sym\-met\-ric'' half-planes separated by an infinite horizontal line of red letters.  
 The property of semi-symmetry means  that the lower half-plane can be obtained from the upper 
 half-plane one by the composition of the usual mirror symmetry and  repainting of some black letters 
 into white ones. In that definition we do not restrict  the number of black letters that are converted in white ones:
 the two half-planes can be exactly symmetric or, conversely, the lower half-plane can consist of only white letters. 
It is instructive to investigate subshifts of ${\cal S}_{\text{semi-mirror}}$ where the mismatch between the upper and lower half-planes
 is under control. In the following two examples we briefly discuss two extremal
cases: a subshift in ${\cal S}_{\text{semi-mirror}}$  where the symmetry can be broken for at most one pair of letters, 
and  another subshift where the symmetry is broken for all black letters except (possibly) one.

\begin{examplebis}{ex:semi-mirror}[semi-mirror shift with very small discrepancy between half-planes]
\label{ex:semi-mirror-bis}
  Let us define a new shift  ${\cal S}_{\text{semi-mirror}}' $ (such that ${\cal S}_{\text{semi-mirror}}' \subset {\cal S}_{\text{semi-mirror}}$) 
 by  adding to the definition of  ${\cal S}_{\text{semi-mirror}}$  another condition saying that the symmetry between 
the two half-planes (above and below the line of red letters) is violated  \emph{for at most one pair of letter}, as shown in Fig~\ref{f:semi-mirror-prim}. 
It is easy to verify that this shift is effective. Exactly the same  argument as in Example~\ref{ex:semi-mirror} implies that ${\cal S}_{\text{semi-mirror}}' $ is also non-sofic.
\end{examplebis}

\begin{figure}

  \centering
  \begin{tikzpicture}[scale=0.25,x=1cm,baseline=2.125cm]

  \pgfmathsetseed{5}
    \foreach \x in {1,...,12} \foreach \y in {1,...,7}
    {
        \pgfmathparse{mod(int(random*23),2) ? "black!10" : "black!66"}
        \edef\colour{\pgfmathresult}
        \path[fill=\colour,draw=black] (\x,7-\y) rectangle ++ (1,1);
        \path[fill=\colour,draw=black] (\x,7+\y) rectangle ++ (1,1);
    }
    
     \foreach \x in {1,...,12} 
    {
        \path[fill=red!41,draw=black] (\x,7) rectangle ++ (1,1);
    }

         \path[fill=black!10,draw=blue,ultra thick] (8,7-4) rectangle ++ (1,1);
 
         \path[fill=black!66,draw=blue,ultra thick] (8,7+4) rectangle ++ (1,1);

\end{tikzpicture}

\caption{A configuration of the shift ${\cal S}_{\text{semi-mirror}}'$ (we show in blue the unique pair of symmetric positions where the upper and lower half-planes mismatch).} \label{f:semi-mirror-prim}
\vspace{-5pt}

\end{figure}

\begin{examplebis2}{ex:semi-mirror}[semi-mirror shift with very large discrepancy between half-planes]
\label{ex:semi-mirror-bis2}
Let us define one more shift  ${\cal S}_{\text{semi-mirror}}'' $ (again ${\cal S}_{\text{semi-mirror}}'' \subset {\cal S}_{\text{semi-mirror}}$) 
by  the condition that the  half-plane  under the horizontal  line of red letters contains \emph{at most one black letter},
see Fig.~\ref{f:semi-mirror-prim2}.
This shift is still effective.
With the new restriction, the argument from  Example~\ref{ex:semi-mirror} does not apply.
Indeed, as we have at most one black letter in the lower half-plane, Kolmogorov complexity of the epitomes ${\cal E}(P)$ defined Example~\ref{ex:semi-mirror} sinks to $O(\log n)$, and we cannot get a contradiction with Proposition~\ref{pbis:epitomes}~(a). 
The fact that our proof of non-soficness fails is not misleading:
it is not hard to verify that ${\cal S}_{\text{semi-mirror}}'' $  is actually sofic. 
We  discuss this example below  in Section~\ref{s:km} in the context of extender sets.
\end{examplebis2}

\begin{figure}[H]

  \centering
  \begin{tikzpicture}[scale=0.25,x=1cm,baseline=2.125cm]

  \pgfmathsetseed{5}
    \foreach \x in {1,...,12} \foreach \y in {1,...,7}
    {
        \pgfmathparse{mod(int(random*23),2) ? "black!10" : "black!66"}
        \edef\colour{\pgfmathresult}
        \path[fill=black!10,draw=black] (\x,7-\y) rectangle ++ (1,1);
        \path[fill=\colour,draw=black] (\x,7+\y) rectangle ++ (1,1);
    }
    
     \foreach \x in {1,...,12} 
    {
        \path[fill=red!41,draw=black] (\x,7) rectangle ++ (1,1);
    }

         \path[fill=black!66,draw=black] (8,7-4) rectangle ++ (1,1);
 
         \path[fill=black!66,draw=blue,ultra thick] (8,7+4) rectangle ++ (1,1);
                                                      
\end{tikzpicture}

\caption{A configuration of the shift ${\cal S}_{\text{semi-mirror}}''$. We show in blue the position in the upper half-plane that is symmetric to the unique black letter in the lower half-plane.} \label{f:semi-mirror-prim2}
\vspace{-5pt}

\end{figure}

In the next example we discuss a shift with very low block complexity, so its non-soficness can be proven only with help of time-bounded Kolmogorov complexity of ordered epitomes.

\begin{example}[semi-mirror shift with low block complexity] 
\label{ex:semi-mirror-sofic}
In this example we use again  subshift ${\cal S}_{0}$  introduced in Notation~\ref{d:s_0},
p.~\pageref{d:s_0}.
We define a new shift ${\cal S}_{\text{semi-mirror}}^{0}$ (such that ${\cal S}_{\text{semi-mirror}}^{0} \subset {\cal S}_{\text{semi-mirror}}'$) by adding to the definition of ${\cal S}_{\text{semi-mirror}}'$ the following requirement: every black-and-white pattern below the horizontal line of red letters must be globally admissible in ${\cal S}_{0}$.

It is easy to verify that ${\cal S}_{\text{semi-mirror}}^{0}$ is also effective. By the construction, for every globally admissible pattern $P$ of size $n\times n$ that can appear below the horizontal line of red letters, we have $\CK(P)  = O(\log n)$. The same property holds for  the patterns that  appear above the horizontal line of red letters, since the mirror symmetry and adding one black letter increases Kolmogorov complexity by at most $O(\log n)$. The same is true for globally admissible patterns involving red letters, since such a pattern can be described by its  parts above and below the horizontal red line, and  both these parts have Kolmogorov complexity $O(\log n)$. Thus, for \emph{every} globally admissible pattern $P$ of size $n\times n$ we have $\CK(P)  = O(\log n)$. Hence,  there are at most $2^{O(\log n)} = \poly(n)$ globally admissible patterns, i.e., the shift has only polynomial block complexity.

The technique of ordered epitomes with the plain Kolmogorov complexity does not apply to this shift since   for every globally admissible pattern $P$ the value of $\CK(P) $ is very low, so  $\CK({\cal E}(P)) $ is also sub-linear, and we cannot obtain a contradiction with Proposition~\ref{pbis:epitomes}~(a).
However, we can prove non-soficness of this shift with help of time bounded Kolmogorov complexity. Indeed, the shift is defined in such a way that for some $n\times n$ patterns that appear below the horizontal line of red  we have  $\CK^{2^{n^3}}({\cal E}(P)) = \Omega(n^{1.5})$. So 
we can repeat the argument from Example~\ref{ex:semi-mirror} with the same o-epitome function ${\cal E}$ (it is  exp-time computable) but this time using Kolmogorov complexity with a time bound $T=2^{\omega(n^2)}$.
By applying Proposition~\ref{pbis:epitomes}~(b), we conclude that the shift is non-sofic.

\end{example}

In the next example we discuss an interesting instance of a shift proposed by Kass and Madden in \cite[Example 2.5]{kass-madden}. We reformulate the proof of non-soficness given in  \cite{kass-madden} in the language of Kolmogorov complexity, in terms of ordered epitomes. In this case, the definition of an epitome function is less straightforward than in the previous examples.

\begin{example}[the shift with no hidden red-black squares]\label{ex:km}
Let $\Sigma$ be the alphabet with three letters (e.g., \emph{black}, \emph{white}, and \emph{red}), and the forbidden patterns be 
all squares  (of all sizes) where the top side consists of red letters, and the bottom one consists of black letters (\emph{hidden red-black squares}), 
as shown in Fig.~\ref{f:012-forbidden}.

\begin{proposition}[\cite{kass-madden}]\label{prop-kass-maden}
The shift on $\mathbb{Z}^2$ defined by the set of forbidden patterns specified above is not sofic.
\end{proposition}
In \cite{kass-madden} this proposition was proven with the technique of \emph{union-increasing chains of extender sets}. In what follows we propose essentially the same  argument, but explain it in terms of ordered epitomes.


\begin{proof}[Proof of Proposition~\ref{prop-kass-maden}] We define for this shift an o-epitome function. 
First of all, we define a class of \emph{simple patterns}: the simple patterns are all square patterns that (i)~consist of only black and white letters (with no red letters), where (ii)~every row  starts with a few successive black letters followed by a sequence of white letters, as show in Fig.~\ref{f:012-standard}.
Every simple pattern of size $n\times n$ can be specified by its \emph{profile} --- a tuple of integers $(k_1,\ldots, k_n)$, where $k_i$ is the number of black letters in the $i$-th row of the pattern. (Thus, a simple pattern with the profile $(k_1,\ldots, k_n)$ is an $n\times n$ square where each $i$-th row starts with $k_i$ black letters followed by $(n-k_i)$ white letters.)

Let epitome ${\cal E}$ assign to each simple pattern its profile, and be undefined for all other patterns. For example, for the pattern $P$ show in Fig.~\ref{f:012-standard} we have ${\cal E}(P)=(4,3,8,5,4,2,4,6)$.

We introduce the natural order $\preccurlyeq$ on the profiles of simple patterns of the same size $n\times n$;  we say that the profile of $P_1$ is \emph{not greater} than the profile of $P_2$, if the first profile is coordinate-wise not greater than the second one. For example, the profiles of the two patterns shown in Fig.~\ref{f:012-patterns} are not greater than the profile of the pattern in Fig.~\ref{f:012-standard} (and incomparable with each other).

The introduced ${\cal E}$ and $\preccurlyeq$ are obviously computable, even in polynomial time. Some work is required to show that $({\cal E}, \preccurlyeq)$  satisfies Definition~\ref{d:ordered-epitome}:
\begin{lemma}\label{l:kass-madden}
The defined above $({\cal E}, \preccurlyeq)$ provides an exp-time  computable o-epitome function for the shift under consideration.
\end{lemma}
\noindent
This lemma is proven implicitly in \cite{kass-madden}. We sketch this proof in \ref{appendix-kass-madden}.

\smallskip

It remains to observe that for every $n$ there are $(n+1)^n$ simple patterns of size $n\times n$ (in each row of a simple pattern the frontier between black and white areas varies between $0$ and $n$). Therefore, for some simple patterns $P$ of size $n\times n$ the Kolmogorov complexity of their profile is greater than $n\log (n+1)$, i.e.,  even the plain Kolmogorov complexity $\CK({\cal E}(P))$ is super-linear. We apply Proposition~\ref{pbis:epitomes}~(a) and conclude that the shift is not sofic.
\end{proof}

\end{example}

\begin{figure}[H]
\centering
\begin{subfigure}[t]{0.23\linewidth}
\begin{center}
  \begin{tikzpicture}[scale=0.27,x=1cm,baseline=2.125cm]
  
  \pgfmathsetseed{1}
    \foreach \x in {0,...,7} \foreach \y in {0,...,7}
    {
        \pgfmathparse{mod(int(random*23),2) ? "red!41" : "black!66"}
        \edef\colour{\pgfmathresult}
        \path[fill=\colour,draw=black] (\x,\y) rectangle ++ (1,1);
    }
    
     \foreach \x in {0,...,7} 
    {
        \path[fill=red!41,draw=black] (\x,7) rectangle ++ (1,1);
    }
     \foreach \x in {0,...,7} 
    {
        \path[fill=black!66,draw=black] (\x,0) rectangle ++ (1,1);
    }

        \path[fill=black!7,draw=black] (3,4) rectangle ++ (1,1);
        \path[fill=black!7,draw=black] (3,2) rectangle ++ (1,1);
        \path[fill=black!7,draw=black] (6,5) rectangle ++ (1,1);
        \path[fill=black!7,draw=black] (1,1) rectangle ++ (1,1);
        \path[fill=black!7,draw=black] (5,3) rectangle ++ (1,1);
        \path[fill=black!7,draw=black] (6,3) rectangle ++ (1,1);
        \path[fill=black!7,draw=black] (5,2) rectangle ++ (1,1);
        \path[fill=black!7,draw=black] (1,6) rectangle ++ (1,1);
        \path[fill=black!7,draw=black] (0,5) rectangle ++ (1,1);
        \path[fill=black!7,draw=black] (2,3) rectangle ++ (1,1);
        \path[fill=black!7,draw=black] (2,1) rectangle ++ (1,1);
        \path[fill=black!7,draw=black] (4,1) rectangle ++ (1,1);
        \path[fill=black!7,draw=black] (5,1) rectangle ++ (1,1);                        
        \path[fill=black!7,draw=black] (0,2) rectangle ++ (1,1);
        \path[fill=black!7,draw=black] (3,6) rectangle ++ (1,1); 
        \path[fill=black!7,draw=black] (4,6) rectangle ++ (1,1); 
        \path[fill=black!7,draw=black] (5,6) rectangle ++ (1,1);       
        \path[fill=black!7,draw=black] (6,6) rectangle ++ (1,1);       
        \path[fill=black!7,draw=black] (7,6) rectangle ++ (1,1);       
        \path[fill=black!7,draw=black] (0,6) rectangle ++ (1,1);       
        \path[fill=black!7,draw=black] (7,1) rectangle ++ (1,1);  
        \path[fill=black!7,draw=black] (6,2) rectangle ++ (1,1);  
        \path[fill=black!7,draw=black] (7,4) rectangle ++ (1,1);  
        \path[fill=black!7,draw=black] (2,5) rectangle ++ (1,1); 
        \path[fill=black!7,draw=black] (0,4) rectangle ++ (1,1); 
        \path[fill=black!7,draw=black] (1,4) rectangle ++ (1,1); 
        \path[fill=black!7,draw=black] (3,1) rectangle ++ (1,1); 
        \path[fill=black!7,draw=black] (5,5) rectangle ++ (1,1); 
        \path[fill=black!7,draw=black] (1,2) rectangle ++ (1,1); 
        \path[fill=black!7,draw=black] (4,4) rectangle ++ (1,1); 
        \path[fill=black!7,draw=black] (5,3) rectangle ++ (1,1); 

        \path[fill=black!66,draw=black] (0,1) rectangle ++ (1,1); 
        \path[fill=black!66,draw=black] (1,5) rectangle ++ (1,1); 
        \path[fill=black!66,draw=black] (6,3) rectangle ++ (1,1); 
        \path[fill=black!66,draw=black] (6,1) rectangle ++ (1,1);
        \path[fill=red!41,draw=black] (2,6) rectangle ++ (1,1);  
        \path[fill=red!41,draw=black] (6,3) rectangle ++ (1,1); 
        \path[fill=red!41,draw=black] (4,2) rectangle ++ (1,1);
   
                                      
\end{tikzpicture}
\end{center}
  \caption{Forbidden pattern: a square with a red top and a black bottom.}
  \label{f:012-forbidden}
\end{subfigure}
\quad
\begin{subfigure}[t]{0.23\linewidth}
\begin{center}
  \begin{tikzpicture}[scale=0.27,x=1cm,baseline=2.125cm]
 
     
    \foreach \x in {0,...,7} \foreach \y in {0,...,7}
    {
          \path[fill=black!7,draw=black] (\x,\y) rectangle ++ (1,1);
    }

 \edef\y{0}
 \edef\xmax{3}
 \foreach \x in {0,...,\xmax}
 \path[fill=black!66,draw=black] (\x,\y) rectangle ++ (1,1);

 \edef\y{1}
 \edef\xmax{2}
 \foreach \x in {0,...,\xmax}
 \path[fill=black!66,draw=black] (\x,\y) rectangle ++ (1,1);

 \edef\y{2}
 \edef\xmax{7}
 \foreach \x in {0,...,\xmax}
 \path[fill=black!66,draw=black] (\x,\y) rectangle ++ (1,1);

 \edef\y{3}
 \edef\xmax{4}
 \foreach \x in {0,...,\xmax}
 \path[fill=black!66,draw=black] (\x,\y) rectangle ++ (1,1);

 \edef\y{4}
 \edef\xmax{3}
 \foreach \x in {0,...,\xmax}
 \path[fill=black!66,draw=black] (\x,\y) rectangle ++ (1,1);

 \edef\y{5}
 \edef\xmax{1}
 \foreach \x in {0,...,\xmax}
 \path[fill=black!66,draw=black] (\x,\y) rectangle ++ (1,1);

 \edef\y{6}
 \edef\xmax{3}
 \foreach \x in {0,...,\xmax}
 \path[fill=black!66,draw=black] (\x,\y) rectangle ++ (1,1);

 \edef\y{7}
 \edef\xmax{5}
 \foreach \x in {0,...,\xmax}
 \path[fill=black!66,draw=black] (\x,\y) rectangle ++ (1,1);

                                         
\end{tikzpicture}
\end{center}
	\caption{A pattern for which the epitome ${\cal E}$ is defined: each row starts with a few black letters on the left followed by white letters on the right.}
 \label{f:012-standard}
\end{subfigure}
\quad
\begin{subfigure}[t]{0.46\linewidth}
\begin{center}
  \begin{tikzpicture}[scale=0.27,x=1cm,baseline=2.125cm]
 
     
    \foreach \x in {0,...,7} \foreach \y in {0,...,7}
    {
          \path[fill=black!7,draw=black] (\x,\y) rectangle ++ (1,1);
    }

 \edef\y{0}
 \edef\xmax{3}
 \foreach \x in {0,...,\xmax}
 \path[fill=black!66,draw=black] (\x,\y) rectangle ++ (1,1);

 \edef\y{1}
 \edef\xmax{1}
 \foreach \x in {0,...,\xmax}
 \path[fill=black!66,draw=black] (\x,\y) rectangle ++ (1,1);
  \path[draw=blue,line width=0.5mm,dashed] (0,\y) rectangle ++ (3,1);

 \edef\y{2}
 \edef\xmax{7}
 \foreach \x in {0,...,\xmax}
 \path[fill=black!66,draw=black] (\x,\y) rectangle ++ (1,1);

 \edef\y{3}
 \edef\xmax{4}
 \foreach \x in {0,...,\xmax}
 \path[fill=black!66,draw=black] (\x,\y) rectangle ++ (1,1);

 \edef\y{4}
 \edef\xmax{2}
 \foreach \x in {0,...,\xmax}
 \path[fill=black!66,draw=black] (\x,\y) rectangle ++ (1,1);
  \path[draw=blue,line width=0.5mm,dashed] (0,\y) rectangle ++ (4,1);

 \edef\y{5}
 \edef\xmax{1}
 \foreach \x in {0,...,\xmax}
 \path[fill=black!66,draw=black] (\x,\y) rectangle ++ (1,1);

 \edef\y{6}
 \edef\xmax{3}
 \foreach \x in {0,...,\xmax}
 \path[fill=black!66,draw=black] (\x,\y) rectangle ++ (1,1);

 \edef\y{7}
 \edef\xmax{3}
 \foreach \x in {0,...,\xmax}
 \path[fill=black!66,draw=black] (\x,\y) rectangle ++ (1,1);
  \path[draw=blue,line width=0.5mm,dashed] (0,\y) rectangle ++ (6,1);


 \end{tikzpicture}
 \quad \begin{tikzpicture}[scale=0.27,x=1cm,baseline=2.125cm]
 
     
    \foreach \x in {0,...,7} \foreach \y in {0,...,7}
    {
          \path[fill=black!7,draw=black] (\x,\y) rectangle ++ (1,1);
    }

 \edef\y{0}
 \edef\xmax{3}
 \foreach \x in {0,...,\xmax}
 \path[fill=black!66,draw=black] (\x,\y) rectangle ++ (1,1);

 \edef\y{1}
 \edef\xmax{2}
 \foreach \x in {0,...,\xmax}
 \path[fill=black!66,draw=black] (\x,\y) rectangle ++ (1,1);

 \edef\y{2}
 \edef\xmax{3}
 \foreach \x in {0,...,\xmax}
 \path[fill=black!66,draw=black] (\x,\y) rectangle ++ (1,1);
  \path[draw=blue,line width=0.5mm,dashed] (0,\y) rectangle ++ (8,1);

 \edef\y{3}
 \edef\xmax{0}
 \foreach \x in {0,...,\xmax}
 \path[fill=black!66,draw=black] (\x,\y) rectangle ++ (1,1);
  \path[draw=blue,line width=0.5mm,dashed] (0,\y) rectangle ++ (5,1);

 \edef\y{4}
 \edef\xmax{3}
 \foreach \x in {0,...,\xmax}
 \path[fill=black!66,draw=black] (\x,\y) rectangle ++ (1,1);

 \edef\y{5}
 \edef\xmax{1}
 \foreach \x in {0,...,\xmax}
 \path[fill=black!66,draw=black] (\x,\y) rectangle ++ (1,1);

 \edef\y{6}
 \edef\xmax{3}
 \foreach \x in {0,...,\xmax}
 \path[fill=black!66,draw=black] (\x,\y) rectangle ++ (1,1);

 \edef\y{7}
 \edef\xmax{5}
 \foreach \x in {0,...,\xmax}
 \path[fill=black!66,draw=black] (\x,\y) rectangle ++ (1,1);

\end{tikzpicture}
\end{center}
\caption{A pair of incomparable patterns.} \label{f:012-patterns}
\end{subfigure}
\vspace{5pt}
\caption{}
\end{figure}

\subsection{The technique of ordered epitomes and union-increasing chains of extender sets}
\label{s:km}

In Example~\ref{ex:km} we translated in the language of epitomes a proof proposed by S.~Kass and K.~Madden in \cite{kass-madden}.
In this section we show that this example is a special case of the general situation:  
the  technique of Kass and Madden is equivalent to a special case of the technique of ordered epitomes.
In this special case we require that the epitome functions have decidable domains, and we use  the plain Kolmogorov complexity (i.e., with no  bounds for computational resources).
We begin with two definitions from \cite{kass-madden}.

\begin{definition} 
Let $\cal S$ in $\mathbb{Z}^2$ be a shift and let $P$ be a pattern  with support $B_n$.
The extender set of $P$, denoted $E_{\cal S}(P)$, consists of all patterns $Q$ on $F_n$ such that the union of $P$ and $Q$ gives
a configuration in $\cal S$.
The family of extender sets for all globally admissible patterns $P$ with support $B_n$  is denoted 
\[
Ext_{\cal S}^n  := \{ E_{\cal S}(P)\ |\ P \text{ is a globally admissible pattern on }B_n\}.
\]
\end{definition}
In some sense, the cardinality of $Ext_{\cal S}^n$ corresponds to the ``information flow'' between an $n\times n$ pattern and the rest of the configuration on $\mathbb{Z}^2$. Intuitively, this ``information flow'' crosses the borderline around an $n\times n$ pattern, while the length of this borderline is only $O(n)$. This intuition can be made more precise for SFTs: for every shift of finite type $\cal S$ on $\mathbb{Z}^2$ the size of $Ext_{\cal S}^n$ cannot be bigger than $2^{O(n)}$. 
However, this property is not necessary true for multi-dimensional sofic shifts. In fact, it is known that for some sofic shifts on $\mathbb{Z}^2$ the size of $Ext_{\cal S}^n$ can grow faster than $2^{O(n)}$ (see Remark~\ref{rem:1}). For instance, for the sofic shift in Example~\ref{ex:semi-mirror-bis2} the size of $Ext_{\cal S}^n$ grows as $2^{\Omega(n^2)}$.

Thus, to show that a shift is not sofic, it is not enough to study the \emph{size} of  the  family of extender sets, we need to reveal subtler properties of these objects. More specifically, we are going to use the structure of the class of extender sets imposed by the relation of inclusion. Technically, we use  union-increasing chains of extender sets.
\begin{definition}
Let ${ S}_1,\ldots,{S}_m$  be a finite sequence of non-empty sets. This sequence is called 
\emph{union-increasing} chain if the partial unions of the ${S}_i$ strictly increase, i.e., 
if for each $1\le i\le m$
\[
 {S}_i  \nsubseteq \bigcup\limits_{j=1}^{i-1} {S}_j .
\]
\end{definition}
Now we can formulate the theorem of  Kass and Madden. 
\begin{theorem}[\cite{kass-madden}]\label{th:km}
Let  $\cal S$ in $\mathbb{Z}^2$ be a shift space. Suppose given any $M > 0$, there exists an $n > 0$, an $m > M^{n}$, 
and globally admissible patterns $P_1,\ldots,P_m$ with support $B_n$ for which 
\[
Ext_{\cal S}(P_1), \ldots,Ext_{\cal S}(P_m)
\]
 is a union-increasing chain. Then $\cal S$ is non-sofic.
\end{theorem}
Loosely speaking, this theorem claims that the ``essential information'' that must cross the borderline of an $n\times n$
pattern (in a sofic shift)
corresponds not just to the number of extender sets but to the length of union-increasing chains of extender sets.
This theorem can be reformulated as follows.
\begin{corollary}\label{cor:km}
Let  $\cal S$ in $\mathbb{Z}^2$ be a shift space. For every $n$ we choose a sequence of globally admissible patterns $P_1^n,\ldots,P_{m(n)}^n$ with support $B_n$ for which the sequence of extender sets
\[
Ext_{\cal S}(P_1^n), \ldots,Ext_{\cal S}(P_{m(n)}^n)
\]
is a union-increasing chain. If $\log m(n) = \omega(n)$, then  $\cal S$ is non-sofic.
\end{corollary}
This approach is pretty general. 
Kass and Madden have shown in particular  that Theorem~\ref{th:km} allows to express any argument based on the technique of 
 Pavlov (\cite[Theorem~1.1]{pavlov}). 

We show that every proof of non-soficness based on Theorem~\ref{th:km} can be rephrased in terms of o-epitome functions from Definition~\ref{d:ordered-epitome} with decidable domains and with the  plain (resource-unbounded) Kolmogorov complexity, and vice-versa. We prove this equivalence in two steps: first we show that an argument in the language of union-increasing chains of extender sets can be translated in the language of o-epitome functions (Proposition~\ref{p:km2e}); then we prove the opposite implication and show that an argument explained in the language of o-epitome functions can be reformulated in terms of union-increasing chains of extender sets  (Proposition~\ref{p:e2km} and Proposition~\ref{p:counting}).

Next proposition is formulated and proven in terms  of families of restricted o-epitomes (let us remind that $({\cal E}_n,\preccurlyeq_n)$ is the restriction of $({\cal E},\preccurlyeq)$ onto patterns of size $n\times n$).
\begin{proposition}\label{p:km2e}
Let  $\cal S$ in $\mathbb{Z}^2$ be a shift space. Assume that for every $n$ there is a sequence of globally admissible patterns $P_1,\ldots,P_{m(n)}$ with support $B_n$ such that  the set of extender sets
\[
Ext_{\cal S}(P_1), \ldots,Ext_{\cal S}(P_{m(n)})
\]
is union-increasing. Then  for every  $n$ there exists an onto mapping
\[
 {\cal E}_n :[\text{patterns with domain } B_n] \to \{1,\ldots, m\}
 \]
defined on a subset of globally admissible patterns in $\cal S$ with support $B_n$ and a partial order $\preccurlyeq_n$ on  $\{1,\ldots, m(n)\}$ such that 
$( {\cal E}_n , \preccurlyeq_n)$ satisfies the definition of a family of restricted o-epitome functions for $\cal S$.
\end{proposition}
\begin{proof}
The construction of a suitable family of epitome functions is straightforward: we let ${\cal E}_n(P_i) = i$ for $i=1,\ldots,m(n)$, and assume that
for all other patterns the function ${\cal E}_n$ is undefined. We introduce  on the values of ${\cal E}_n$ the linear order that is inverse to the natural order on the integer numbers,
\[
m(n)\preccurlyeq_n   \ldots  \preccurlyeq_n 2  \preccurlyeq_n 1.
\]
Let us verify that the obtained $({\cal E}_n,\preccurlyeq_n)$ satisfies the definition of a family of o-epitome functions. 
Since the sequence of extender sets $Ext_{\cal S}(P_i)$ is union-increasing, for every $i$ we have
\[
Ext_{\cal S}(P_i)  \not\subset  \bigcup\limits_{j=1}^{i-1} Ext_{\cal S}(P_i).
\]
This means that there exists a pattern $R$ with support $F_n$ that is compatible with $P_i$ (i.e., belongs to $Ext_{\cal S}(P_i)$) but not with $P_1,\ldots, P_{i-1}$ (i.e., does not belong to $ \bigcup\limits_{j=1}^{i-1} Ext_{\cal S}(P_i)$).
Thus, for every $n\times n$ pattern $P'$  compatible $R$, if ${\cal E}_n(P')$ is defined then
${\cal E}_n(P')\preccurlyeq_n i$, which is exactly the definition of restricted o-epitome function.
\end{proof}

Assume that for every positive integer number $n$ we have a sequence of globally admissible patterns  $P^n_1,\ldots,P^n_{m(n)}$ with support $B_n$ such that  the corresponding sequence  of extender sets
\[
Ext_{\cal S}(P^n_1), \ldots,Ext_{\cal S}(P^n_{m(n)})
\]
is a union-increasing chain. By applying Proposition~\ref{p:km2e} we obtain a family of restricted o-epitome functions ${\cal E}_n$ that have  $m(n)$ values for patterns of size $n\times n$. In general, this  epitome function can be non computable. However, due to a trivial reason it is computable with the oracle $\cal O$ that contains a description of the sequences $P^n_1,\ldots,P^n_{m(n)}$.

If the number $\log m(n)$ grows faster than a linear function (i.e., $m(n) =2^{\omega(n)}$, which is necessary to apply Corollary~\ref{cor:km}), then from the counting argument it follows that Kolmogorov complexity of ${\cal E}_n(P^n_i)$ (even with oracle $\cal O$) is not bounded by a linear function. More precisely, for every $\lambda>0$ and for a large enough $n$ there is a pattern $P^n_i$ such that $\CK^{\cal O}({\cal E}_n(P^n_i))  > \lambda n$. So we can apply Proposition~\ref{pbis:epitomes}(c) and conclude that $\cal S$ is non-sofic.

Thus, if we can prove that a shift $\cal S$ is non-sofic with the technique of Kass and Madden (Corollary~\ref{cor:km}), then due to  Proposition~\ref{p:km2e} we can apply Proposition~\ref{pbis:epitomes}~(c)  (with Kolmogorov complexity relativized conditional on the chosen above oracle $\cal O$) 
 and prove non-soficness of $\cal S$ with the technique of epitomes.
 
 \medskip
 
 In the next two propositions we show that the reverse translation (from the language of ordered epitomes to the language of chains of extender sets) is also possible as long as we use o-epitome functions with decidable domains and resource-unbounded Kolmogorov complexity.   
 
\begin{proposition}\label{p:e2km}
Let  $\cal S$ be a shift space,  $({\cal E}_n,\preccurlyeq_n)$ be a family of restricted o-epitome functions for $\cal S$, 
and  $m=m(n)$ be the number of images of ${\cal E}_n$ computed for globally admissible patterns. 
Then there exists a family of globally admissible  patterns $P_1,\ldots, P_m$  such that the chain of extender sets
\[
Ext_{\cal S}(P_1), \ldots,Ext_{\cal S}(P_m)
\]
is  union-increasing.
\end{proposition}
\begin{proof}
Denote by $e_1,\ldots,e_m$ the images of ${\cal E}_n$. Without loss of generality we may assume that for every $i$
\begin{equation} \label{eq:km}
e_j \not\preccurlyeq_n e_i \text{ for all } j<i 
\end{equation}
(in every finite  partially ordered set we can arrange the indices of $e_j$ so that \eqref{eq:km} holds true: $e_1$ should be one of the maximal elements, $e_2$ must be maximal among the remaining elements, and so on). Then, we fix  globally admissible patterns 
$P_i$ with support $B_n$ such that ${\cal E}_n(P_i) = e_i$. By the definition of an epitome function, for each $i$ there exists a pattern
$R_i$ with support $F_n$ that is compatible with $P_i$ but not compatible with $P_1,\ldots, P_{i-1}$. 
Therefore, $Ext_{\cal S}(P_i)$ is not contained in the union $\bigcup\limits_{j=1}^{i-1} Ext_{\cal S}(P_i)$. Thus, the sequence 
of extender sets $Ext_{\cal S}(P_i)$ is union-increasing, and we are done.
\end{proof}

\begin{proposition}\label{p:counting}
Let $({\cal E},\preccurlyeq)$ be a computable o-epitome function for an effective shift $\cal S$. 

(a) Assume that for  an admissible pattern $P$ of size $n\times n$
\begin{equation}
\label{eq:ck}
\CK({\cal E}(P)) = m,
\end{equation}
and $m\gg \log n$. Then there exist $2^{\Omega(m)}$ different values of ${\cal E}$ corresponding to  patterns of size $n\times n$
(globally admissible or non-admissible).

(b) Assume the domains of ${\cal E}$  are decidable (there is an algorithm that can decide for every $n$ and every pattern $P$ of size $n\times n$ whether  ${\cal E}(P)$ is defined or not)
and that \eqref{eq:ck} with  $m\gg \log n$ holds for at least one globally admissible pattern $P$ of size $n\times n$.
Then there exist $2^{\Omega(m)}$ values of ${\cal E}$ corresponding to globally admissible patterns of this size.

(c) Both propositions (a) and (b) relativize: they remain true for Kolmogorov complexity relativized with any oracle $\cal O$.
\end{proposition}
\begin{proof}
(a)   For a given $n$, we can apply the algorithm computing  ${\cal E}$  to each pattern of size $n\times n$ and run these computations in parallel.  As some of these computations converge, we can write down the revealed values of ${\cal E}$ one by one.
(If the domain of ${\cal E}$ is undecidable, we never know whether the last terminating computation is over,  and some of these computations last infinitely long time without any result.)
Let $N_n$ be the number of all values that eventually appear. 
Each value of ${\cal E}$ can be specified by its index in this list (in the order of appearance)  together with the binary expansions of $n$. The index (position in the list) can be specified by $\lceil \log N_n \rceil$ bits.
Therefore,   for every $n\times n$ pattern $P$ we have
 \[
 \CK({\cal E}(P))  = O(\log N_n) +O (\log n) , 
 \]
and the statement~(a) follows.

(b) Let $N_n$ be the number of values of ${\cal E}(P)$ for all patterns $P$ of size $n\times n$,
and let $N'_n$ be the number of values of ${\cal E}(P)$ for \emph{globally admissible} patterns $P$ of this size.
Denote $\ell = \lceil \log N_n'\rceil$.

If $N'_n > N_n/2$, the statement~(b) follows from~(a). Otherwise, we compute ${\cal E}(P)$ for all patterns of size $n\times n$ (this is possible since the domain of an epitome function is decidable)
and  start enumerating non-admissible patterns of the shift. As non-admissible patterns are progressively revealed, 
we eliminate one by one the values of ${\cal E}(P)$ that do not correspond to any non-discarded  (and potentially globally admissible) pattern $P$. We stop this process
when there remain exactly $2^\ell$ non-eliminated values of epitomes (that can \emph{a priori} represent globally admissible patterns). 
Now the  epitome of  every globally admissible $P$ can be specified by the numbers $n$ and $\ell$ and by the ordinal number  of the value ${\cal E}(P)$ in the list of  $2^\ell$ non-eliminated candidates. Such an ordinal number can be represented by a string of  $\ell_n$ bits. Thus,
\[
 \CK({\cal E}(P))  = O(\ell)+ O (\log n), 
\]
and we are done.

(c) It is easy to see that the argument remains valid for algorithms accessing an oracle.
\end{proof}

Assume that  we can prove non-soficness of a shift $\cal S$  with the technique of ordered epitomes from Proposition~\ref{pbis:epitomes} using Kolmogorov complexity without time bounds. This is possible if we have for $\cal S$ an o-epitome function $({\cal E},\preccurlyeq)$ that is computable (possibly with some oracle $\cal O$) and such that $\CK({\cal E}(P))$ is not bounded by a linear function. Assume now that the union of domains of ${\cal E}$ is a decidable set.
Then it follows from Proposition~\ref{p:counting}~(b)  that  the number of values of ${\cal E}$ for globally admissible patterns grows as $2^{\omega(n)}$. 
Now we apply Proposition~\ref{p:e2km} and obtain a list of globally admissible $n \times n$ patterns $P^n_1,\ldots, P^n_{m(n)}$  
with  $m(n) = 2^{\omega(n)}$ such that the chain of extender sets
$
Ext_{\cal S}(P^n_1), \ldots,Ext_{\cal S}(P^n_{m(n)})
$
is union-increasing.
Hence, we can apply Corollary~\ref{cor:km} and prove non-soficness of $\cal S$  with the argument of Kass and Madden.
 
\medskip

Proposition~\ref{p:km2e} and Proposition~\ref{p:e2km}  show that the technique of union-increasing chains of extender sets is  closely connected with  the version of the technique of ordered epitomes with a decidable domain and the plain Kolmogorov complexity.
On the other hand, it seems that an argument using plain epitome functions with undecidable domains 
(see Example~\ref{ex:busy-beaver}) cannot be translated directly in the language of union-increasing chains of extender sets.
We also recall that the method of epitomes is getting stronger when we use  resource-bounded version of Kolmogorov complexity.
In particular, it helps to handle shifts with sub-exponential block complexity, which is unachievable with Theorem~\ref{th:km} and Corollary~\ref{cor:km}.

\section{Conclusion}

To the best of our knowledge, all previously known proofs of non-soficness for multi-dimension effective shifts (including the one proposed in \cite{pavlov}) can be reformulated in terms of union-increasing chains of extender sets, as proposed in \cite{kass-madden}. Therefore, these proofs can be rephrased in terms of ordered epitomes with the plain Kolmogorov complexity, as we explained in Section~\ref{s:km}.
The technique of epitomes becomes more powerful if we use resource-bounded Kolmogorov complexity; in particular, it allows to handle shifts with sub-exponential block complexity. So it seems that Kolmogorov complexity helps to capture substantial properties that are \emph{necessary} for soficness. It would be interesting to adapt the technique of resource-bounded Kolmogorov complexity to reveal properties that are \emph{sufficient} for a shift to be sofic.

\begin{problem*}
Find {sufficient} conditions of soficness  that can be formulated in terms of resource-bounded Kolmogorov
complexity.
\end{problem*}


\paragraph{Acknowledgments} We thank  Bruno Durand, Alexander Shen, and Ilkka T\"orm\"a for fruitful  discussions.
We are grateful to  Pierre Guillon  and Emmanuel Jeandel  for motivating comments. We also thank the anonymous referees
of STACS~2019 and JCSS  for many insightful comments.

In the  source of  Fig.~\ref{f:sft2sofic} we use a Tikz code based on examples by Marco Miani on  \texttt{texample.net} and Martin Wegmann on \texttt{en.wikipedia.org/wiki/PGF/TikZ}.

\bibliography{nonsofic.bib}

\begin{thebibliography}{10}

\bibitem{mirror}
N.~Aubrun, S.~Barbieri, and E.~Jeandel.
\newblock About the domino problem for subshifts on groups. {S}equences,
  groups, and number theory.
\newblock {\em Birkhäuser, Cham}, pages 331--389, 2018.

\bibitem{dls}
B.~Durand, L.~Levin, and A.~Shen.
\newblock Complex tilings.
\newblock {\em The Journal of Symbolic Logic}, 73:2:593--613, 2008.

\bibitem{hennie-stearns}
F.~C. Hennie and R.~E. Stearns.
\newblock Two-tape simulation of multitape {T}uring machines.
\newblock {\em Journal of the ACM}, 13(4):533--546, 1966.

\bibitem{hochman2009dynamics}
M.~Hochman.
\newblock On the dynamics and recursive properties of multidimensional symbolic
  systems.
\newblock {\em Inventiones mathematicae}, 176(1):131, 2009.

\bibitem{jeandel-hdr}
E.~Jeandel.
\newblock Propriétés structurelles et calculatoires des pavages.
\newblock {\em Habilitation thesis, Universit\'e Montpellier~2}, 2011.

\bibitem{kass-madden}
S.~Kass and K.~Madden.
\newblock A sufficient condition for non-soficness of higher-dimensional
  subshifts.
\newblock {\em Proceedings of the American Mathematical Society},
  141:11:3803--3816, 2013.

\bibitem{kolmogorov-three-approaches}
A.~N. Kolmogorov.
\newblock Three approaches to the quantitative definition of information.
\newblock {\em Problems of information transmission}, 1:1:1--7, 1965.

\bibitem{li-vitanyi-book}
M.~Li and P.~Vit\'anyi.
\newblock An introduction to {K}olmogorov complexity and its applications.
\newblock {\em 3rd ed. Springer, New York}, 2008.

\bibitem{sft}
D.~Lind and B.~Marcus.
\newblock An introduction to symbolic dynamics and coding.
\newblock {\em Cambridge University Press}, 1995.

\bibitem{symbolic-dynamics-1938}
M.~Morse and G.~A. Hedlund.
\newblock Symbolic dynamics.
\newblock {\em Amer. J. Math.}, 60:815--866, 1938.

\bibitem{symbolic-dynamics-1940}
M.~Morse and G.~A. Hedlund.
\newblock Symbolic dynamics {II}: {S}turmian trajectories.
\newblock {\em Amer. J. Math.}, 62:1--42, 1940.

\bibitem{mozes}
S.~Mozes.
\newblock Tilings, substitution systems and dynamical systems generated by
  them.
\newblock {\em Journal d'analyse math\'ematique (Jerusalem)}, 53:139--186,
  1989.

\bibitem{nies2009computability}
Andr{\'e} Nies.
\newblock {\em Computability and randomness}, volume~51.
\newblock OUP Oxford, 2009.

\bibitem{ormes-pavolov}
N.~Ormes and R.~Pavlov.
\newblock Extender sets and multidimensional subshifts.
\newblock {\em Ergodic Theory and Dynamical Systems}, 36:3:908--923, 2016.

\bibitem{pavlov}
R.~Pavlov.
\newblock A class of nonsofic multidimensional shift spaces.
\newblock {\em Proceedings of the American Mathematical Society},
  141:3:987--996, 2013.

\bibitem{rogers1987theory}
Hartley Rogers~Jr.
\newblock {\em Theory of recursive functions and effective computability}.
\newblock MIT Press, 1987.

\bibitem{rumyantsev-ushakov}
A.~Rumyantsev and M.~Ushakov.
\newblock Forbidden substrings, {K}olmogorov complexity and almost periodic
  sequences.
\newblock {\em In Proc. Annual Symposium on Theoretical Aspects of Computer
  Science}, pages 396--407, 2006.

\bibitem{salo}
V.~Salo.
\newblock Subshifts with sparse projective subdynamics.
\newblock {\em arXiv preprint}, arXiv:1605.09623, 2016.

\bibitem{shen-vereshchagin-book}
Alexander Shen, Vladimir~A. Uspensky, and Nikolay Vereshchagin.
\newblock {\em Kolmogorov complexity and algorithmic randomness}, volume 220.
\newblock American Mathematical Soc., 2017.

\bibitem{weiss}
B.~Weiss.
\newblock Subshifts of finite type and sofic systems.
\newblock {\em Monatsh. Math.}, 77:462--474, 1973.

\bibitem{westrick}
L.~B. Westrick.
\newblock Seas of squares with sizes from a ${\Pi}^0_1 $ set.
\newblock {\em Israel Journal of Mathematics}, 22:1, 2017.

\end{thebibliography}

\appendix

\section{Proof of Theorem~\ref{th:dls}}\label{appendix-dls}

It is enough to prove the theorem for shifts of finite type (since a configuration in a sofic shift is a coordinate-wise projection  of a configuration from a shift of finite type). 

We fix an SFT (which can be defined by a set of forbidden patterns $F$ where every pattern in $F$ is a pair of neighboring letters) and show that for every $k$ there exists a locally admissible $(2^{k}+1) \times (2^{k}+1)$-pattern $P_k$ with the required property:
every $n\times n$ square pattern inside $P_k$ has Kolmogorov complexity $O(n)$.

To this end we choose an arbitrary coloring of the border line of a square of size $(2^{k}+1) \times (2^{k}+1)$ so that this coloring can be extended 
in at least one way to a locally admissible coloring of the entire  square, see Fig.~\ref{f:recursuve-construction}~(a). 
Since the set of all colorings of the square is finite,  we can find algorithmically
(by a naive brute-force search) the lexicographically first coloring
of the horizontal and vertical centerlines of the square that are compatible with the fixed coloring of the border line (i.e., the coloring of the border line together
with coloring of the centerlines can be extended to a locally admissible coloring of the whole square),  see Fig.~\ref{f:recursuve-construction}~(b). Observe that the centerlines split the square into four  squares of size $(2^{k-1}+1) \times (2^{k-1}+1)$.

\begin{figure}[H]
\vspace{5pt}
\begin{center}
  \begin{tikzpicture}[scale=0.14,x=1cm,baseline=2.125cm]
 
 \begin{scope}[shift={(0,0)}]
    \foreach \x in {1,...,17} \foreach \y in {1,...,17}
    {
        \path[draw=black] (\x,\y) rectangle ++ (1,1);
    }

   \foreach \x in {1,...,17} \foreach \y in {1,17}
    {
        \path[fill=blue!60,draw=black] (\x,\y) rectangle ++ (1,1);
        \path[fill=blue!60,draw=black] (\y,\x) rectangle ++ (1,1);
    }
                      
\end{scope}

 \begin{scope}[shift={(27,0)}]
 
    \foreach \x in {1,...,17} \foreach \y in {1,...,17}
    {
        \path[draw=black] (\x,\y) rectangle ++ (1,1);
    }

   \foreach \x in {1,...,17} \foreach \y in {1,9,17}
    {
        \path[fill=blue!60,draw=black] (\x,\y) rectangle ++ (1,1);
        \path[fill=blue!60,draw=black] (\y,\x) rectangle ++ (1,1);
    }

\end{scope}
 \begin{scope}[shift={(54,0)}]
 
    \foreach \x in {1,...,17} \foreach \y in {1,...,17}
    {
        \path[draw=black] (\x,\y) rectangle ++ (1,1);
    }

   \foreach \x in {1,...,17} \foreach \y in {1,5,9,13,17}
    {
        \path[fill=blue!60,draw=black] (\x,\y) rectangle ++ (1,1);
        \path[fill=blue!60,draw=black] (\y,\x) rectangle ++ (1,1);
    }

  \end{scope}

  \node at (9.5,-2) {(a)};
  \node at (36.5,-2) {(b)};
  \node at (63.5,-2) {(c)};  
 
   \node at (23,9)  [thick, font=\fontsize{20}{20}\selectfont, thick]  {$\Rightarrow$ };
   \node at (50,9)  [thick, font=\fontsize{20}{20}\selectfont, thick]  {$\Rightarrow$ };
   \node at (77,9)  [thick, font=\fontsize{20}{20}\selectfont, thick]  {$\Rightarrow$ };
   \node at (83,9)  [thick, font=\fontsize{20}{20}\selectfont, thick]  {...}; 
              
\end{tikzpicture}

\caption{Three steps of the recursive coloring procedures for a pattern of size $(2^n+1)\times (2^n+1)$.} \label{f:recursuve-construction}
\end{center}
\end{figure}

Then, we repeat the same procedure recursively: in each of the four squares of size  $(2^{k-1}+1) \times (2^{k-1}+1)$ we find the lexicographically first coloring
of their  horizontal and vertical centerlines that are compatible with the coloring of the border around each of these squares (see Fig.~\ref{f:recursuve-construction}~(c)), and so on. 
On the last step we end up with the lexicographically first valid coloring of isolated patterns of size $1\times 1$ (that must be coherent with the chosen above coloring of the neighborhood). This concludes the construction of $P_k$.

In the procedure  explained above, the initial square of size  $(2^{k}+1) \times (2^{k}+1)$ is split into four  squares of size  $(2^{k-1}+1) \times (2^{k-1}+1)$, 
then into sixteen   squares of size  $(2^{k-2}+1) \times (2^{k-2}+1)$, etc. This recursive procedure  finds the coloring of each of these
squares in ``standard'' positions given (as a kind of boundary condition) the coloring of the border around this square. 

Thus, to find the assigned coloring of a square of size    $(2^{m}+1) \times (2^{m}+1)$ in a ``standard'' position,  the recursive procedure needs to know 
only the coloring of the border around this square (which requires $O(2^{m})$ letters and therefore $O(2^{m})$ bits of information);
it is not hard to see that   the recursive call runs in time

\[
 2^{O\big((2^{m}+1) \times (2^{m}+1)\big)}+   4\times 2^{O\big((2^{m-1}+1) \times (2^{m-1}+1)\big)}+ 16\times 2^{O\big((2^{m-2}+1) \times (2^{m-2}+1)\big)}+\ldots
=  2^{O(2^{2m})}
\]
(this is a recursive computation of depth $m$,  with a brute-force search and four recursive calls on each level of the hierarchy).

An arbitrary square of size $n\times n$ (possibly in a non-standard position) is covered by at most four ``standard'' squares (of size at most twice bigger than
$n$), see Fig.~\ref{f:recursuve-construction-covering}.

\begin{figure}[H]
\begin{center}
  \begin{tikzpicture}[scale=0.20,x=1cm,baseline=2.125cm]

 \begin{scope}[shift={(27,0)}]
 
     \foreach \x in {1,...,17} \foreach \y in {1,...,17}
    {
        \path[draw=black] (\x,\y) rectangle ++ (1,1);
    }

   \foreach \x in {1,...,17} \foreach \y in {1,9,17}
    {
        \path[fill=blue!60,draw=black] (\x,\y) rectangle ++ (1,1);
        \path[fill=blue!60,draw=black] (\y,\x) rectangle ++ (1,1);
    }

    \draw [pattern=crosshatch, pattern color=red!90, thick, dashed]  (7,6) rectangle (14,13) ;
\end{scope}

\end{tikzpicture}

\caption{A pattern of size $7\times 7$ (hatched in red) covered by a quadruple of ``standard'' blocks of size $9\times 9$.} \label{f:recursuve-construction-covering}
\end{center}
	\vspace{-10pt}
\end{figure}

Thus, to identify an $n\times n$-pattern inside $P_k$,  it is enough to describe the quadruple of standard squares 
covering this pattern and, in addition, the position (the coordinates of the corners) of this pattern with respect to the covering 
standard squares.  In turn, each standard square is, by construction, determined by its border line.
Hence, every $n\times n$ square in $P_k$ can be specified with only $O(n)$ bits of information;
moreover, it can be recovered from this description in time $2^{O(n^2)}$  (at first we reconstruct
the four standard squares from their border lines, and then cut out the pattern with the given coordinates 
of the corners).

To conclude the proof of the theorem, we observe that $P_k$ are defined for arbitrarily large $k$, and the compactness argument gives a valid coloring
of the entire $\mathbb{Z}^2$ with the required property.

\section{Proof of Theorem~\ref{th:deep-shift}}\label{appendix-deep-shift}

The high level scheme of the proof of Theorem~\ref{th:deep-shift} is similar to the proof of Theorem~\ref{th:deep-shift-simple}. 
At first we construct ``standard'' building blocks of growing size so that all patterns in these blocks have a large time-bounded
Kolmogorov complexity. Then we define a shift as the closure of these standard patterns.

\underline{Stage 1: building standard blocks}.
We fix parameters $n_i$, $i=0,1,2,\ldots$
\begin{equation}
\label{eq:def-n_0}
   n_0 \gg 1,\  n_{i+1} = (n_0\cdot \ldots \cdot n_i)^c
\end{equation}
(the constant $c$ to be specified later)
 and 
\[
 N_ i :=  n_0\cdot \ldots \cdot n_i.
\]

 In what follows we define inductively  \emph{standard blocks} $Q_i^j$ of size $N_i\times N_i$
 for $i=0,1,2,\ldots$ The major difference between this proof and  the proof of Theorem~\ref{th:deep-shift-simple} is that for each $i$ we define  a set  of $\ell_i = n_{i+1}^2$ standard blocks of level $i$
\[
   Q_i^1, \ldots, Q_i^{\ell_i}
\]
(instead of only two of standard patterns used in the previous proof).  
The construction is hierarchical: each standard block of rank $i$ is defined as an $n_i\times n_i$ array of standard blocks of level $(i-1)$.
 In this inductive construction, we maintain for each $i$ the following property:
 
  \smallskip
  
 \noindent
 \emph{Main Property}:   For every $i>0$,  for the 
 list of standard blocks $Q_i^{1}, \ldots, Q_i^{\ell_i}$, we have that
\begin{equation}
\label{eq:main-p}
  \CK^{t_i} (Q_i^{1}, \ldots, Q_i^{\ell_i} ) \ge  \ell_i \log (\ell_{i-1}!)  = (n_{i+1})^2  \cdot \log \big((n_i^2)!\big)
\end{equation}
 (the time threshold $t_i$ to be specified later).
 
 \smallskip
 
 Our construction is explicit, and the standard blocks $Q_i^{j}$ are computable (given $i$ and $j$) though the time required for this computation
 can be pretty large.

 \begin{remark}
 Let us explain the intuition behind this construction. 
 We inductively construct standard blocks that have the following two properties.
 On the one hand, 
 all large enough patterns in each standard $N_i\times N_i$ blocks should have a high  time-bounded Kolmogorov complexity
 due to the ``local structure'': each of the standard blocks of rank $(i-1)$ (of size  $N_{i-1}\times N_{i-1}$) is complex,
 and these blocks  are independent of each other.
On the other hand, the whole standard block of size  $N_i\times N_i$ should have a high time bounded Kolmogorov complexity (with even a bigger time bound) due to the information embedded in its ``global structure'' (how the  blocks of level $(i-1)$ are arranged inside of a block of level $i$).  
To carry out the inductive argument, we construct for each $i$ a large  family of  standard blocks of size  $N_i\times N_i$, and
these blocks must be in some sense ``independent.''
 \end{remark}

\emph{Base of induction:}
We  defined blocks $Q_0^{1}, \ldots, Q_0^{\ell_0} $  as arbitrary distinct binary matrices of size $n_0\times n_0$
(we can do it assuming that $2^{n_0^2}> \ell_0 = n_1^2$). 
\smallskip

\emph{Inductive step:} Given a set of standard patterns $Q_i^{1}, \ldots, Q_i^{\ell_i} $ of size $N_i\times N_i$, we construct the
standard patterns of the next level, $Q_{i+1}^{1}, \ldots, Q_{i+1}^{\ell_i} $ of size $N_{i+1}\times N_{i+1}$. Each new pattern $Q_{i+1}^{j}$
is defined as an $n_{i+1}\times n_{i+1}$ array composed of  blocks $Q_i^{1}, \ldots, Q_i^{\ell_i}$. 

Since $\ell_i = n_{i+1}^2$, we may require that each block   $Q_{i}^{k}$ is used  exactly once in  every $Q_{i+1}^{j}$. In other words,
every standard block  $Q_{i+1}^{j}$  can be represented by  a permutation 
 \[
   \pi_ j \ : \ \{1,\ldots, \ell_i\} \to \{1,\ldots, \ell_i\} 
 \]
that arranges the set of  standard patterns of the previous level.

There are $(\ell_i \,!)$ possibilities to choose each  permutations $\pi_ j$ 
and $(\, \ell_i\,! \, )^{\ell_{i+1}}$  possibilities to choose all permutations  $ \pi_ 1 ,\ldots,  \pi_ {\ell_{i+1}}$. 
From a trivial counting argument it follows that there exists a tuple of permutations $ \pi_ 1 ,\ldots,  \pi_ {\ell_{i+1}}$ such that
\[
  \CK( \pi_ 1 ,\ldots,  \pi_ {\ell_{i+1}}) \ge \ell_{i+1} \log (\ell_i!)
\]
Therefore, for every computable threshold $h_i$, we can algorithmically find a tuple of permutations such that
 \[
  \CK^{h_{i+1}}( \pi_ 1 ,\ldots,  \pi_ {\ell_{i+1}}) \ge \ell_{i+1} \log (\ell_i!).
 \]
Since each $\pi_j$ is uniquely determined by the corresponding pattern $Q_{i+1}^j$, we obtain
\[
  \CK^{h_{i+1}-\mathrm{gap}_{i+1}}( Q_{i+1}^{1} ,\ldots,  Q_{i+1}^{\ell_{i+1}}) \ge \ell_{i+1} \log (\ell_i!),
\]
where $\mathrm{gap}_{i+1}$ is the time required to extract the permutations $\pi_j$ from the corresponding patterns  $Q_{i+1}^j$. 

Thus, if the threshold $h_{i+1}$ is much bigger than $t_i$, we conclude that 
\[
  \CK^{t_i}( Q_{i+1}^{1} ,\ldots,  Q_{i+1}^{\ell_{i+1}}) \ge \ell_{i+1} \log (\ell_i!),
\]
 i.e., the \emph{Main Property} holds true for the level $i+1$.

\begin{lemma}\label{l:1}
For every $i>0$ and every tuple of $k$ (pairwise different) standard blocks $Q_{i}^{{j_1}} ,\ldots,  Q_{i}^{j_k}$, we have that
\[
 \CK^{t_i'}(Q_{i}^{{j_1}} ,\ldots,  Q_{i}^{j_k}) \ge \frac12 k \log (\ell_{i-1}!)
\]
(the threshold $t_i'$ is specified in what follows). 
\end{lemma}
\begin{remark}
The factor of $\frac12$ in this lemma could be changed to any constant less than $1$.
\end{remark}
\begin{proof}[Proof of lemma] Assume for the sake of contradiction that for some blocks $Q_{i}^{{j_1}} ,\ldots,  Q_{i}^{j_k}$ we have
\[
 \CK^{t_i'}(Q_{i}^{{j_1}} ,\ldots,  Q_{i}^{j_k}) < \frac12 k \log (\ell_{i-1}!)
\]
Then we can provide the following description of the entire list of standard blocks of level~$i$:
 \begin{itemize}
 \item a description of size  $\frac12 k \log (\ell_{i-1}!)$ for these particular $k$ blocks,
 \item $k\log \ell_{i}$ bits to specify the indices $j_1,\ldots,j_k$ of these particular $k$ block,
 \item a straightforward  description  for the other standard blocks of level $i$, which requires  $(\ell_{i}-k)\log (\ell_{i-1}!)$ bits.
 \end{itemize}
We need to add an overhead of size $O(\log k+\log i )$ to join these  parts in one description.
It remains to observe that the summarized length of these data is less than $ \ell_{i+1} \log (\ell_i!)$ 
(the right-hand side of \eqref{eq:main-p}).

Let us estimate the time required to produce the list of all standard blocks of level~$i$ using this description. First of all, we 
compute the list of all standard blocks of lower rank $Q_{i-1}^j$ (we have to assume that the time  required to compute these blocks is much less 
than~$t_i$). Then, in time $t_i'$ we  obtain the $k$ ``peculiar'' blocks with a short description.  With $\poly(N_{i})$ steps  we generate 
the remaining  $(\ell_i -k)$ standard blocks, and in $\poly(N_{i})$ more steps we merge all blocks  in one list.
Assuming that $t_i$ is large enough, we get a contradiction with the \emph{Main Property} of standard blocks.
\end{proof}

\begin{lemma}\label{l:2}
For every standard block of rank $i$,  the plain Kolmogorov complexity of $Q_i^j$ is very low:
\[
 \CK(Q_i^j) = O(\log N_i).
\]

\end{lemma}
\begin{proof}
There are $\ell_i$ standard blocks of level $i$, and the list of all standard blocks can be computed given $i$. Therefore,
\[
\CK(Q_i^j) = O(\log \ell_i).
\]
For the chosen parameters we have $\log \ell_i = \log (n_{i+1}^2) = O(\log N_i)$.
\end{proof}
Observe that the construction of standard blocks is explicit, and we can choose a computable function $\hat T(n)$ so that  all standard blocks are computed given $i$ in  time $\hat T(n)$. This observation implies the statement of Remark~\ref{r:improved-upper-bound}.

\smallskip

\underline{Stage 2: complexity of patterns inside standard blocks}.
Let $P$ be an array of $m\times m$ (pairwise distinct) standard  blocks of level $i$. 
Observe that the size of $P$ (measured in individual letters) is $k\times k$, where $k=mN_i$.

Due to Lemma~\ref{l:1},  we have 
\[
 \CK^{t_i'}(P) \ge \frac12 m^2 \log (\ell_{i-1}!) = \Omega(m^2 \ell_{i-1}) = \Omega\big((mn_i)^2\big).
\]
In other words, the time-bounded Kolmogorov complexity of this pattern of size $k\times k$ is
bigger than 
\[
 \Omega\big((mn_i)^2\big) =  \Omega\big(k^2 / (n_0\cdot \ldots \cdot n_{i-1})^2\big).
\]

\smallskip

\underline{Stage 3: the closure of standard blocks.}
We define a shift as the closure of the standard blocks: a finite pattern  is globally admissible if and only if
it appears in a $2\times 2$ array composed of standard blocks (of some level $i$).
Due to the hierarchical structure if standard blocks, we can conclude that an $N_i\times N_i$ pattern is \emph{not} globally admissible if
it never appears in $2\times 2$ arrays composed of standard blocks of rank $i$.

Let $P$ be a globally admissible pattern of size $k\times k$ in the shift defined above. Let $2N_i \le k < 2N_{i+1}$ for some $i>0$.
Then $P$ must be covered by a $2\times 2$ array of standard blocks of level $i+1$. Therefore, a constant fraction of  $P$
can be represented as an array of blocks of level $i$ (inside a standard block of level $(i+1)$). 
We can apply the bound from Stage 2 and conclude that
\begin{equation}
\label{eq:lower-bound}
 \CK^{T_i}(P) =  \Omega(k^2/(n_0\cdot \ldots \cdot n_{i-1})^2)
\end{equation}
(assuming that the gap between $t'_i$ and $T_i$ is large enough).

On the other hand, the plain Kolmogorov complexity of every globally admissible block is very low. Indeed, to describe a globally admissible 
block $P$ of size $k$, we need to specify four standard blocks covering $P$ and the position of $P$ with respect to the 
grid of standard blocks. Due to Lemma~\ref{l:2}, we obtain
\[
 \CK(P) =  O(\log k). 
\]

\smallskip

\underline{Stage 4: the choice of parameters}.
We chose constants $c$ in \eqref{eq:def-n_0} so that for $k=\Omega(N_i)$ equality \eqref{eq:lower-bound} rewrites to 
\[
 \CK^{T_i}(P) =  \Omega(k^{2-\epsilon}).
\]

It remains to comment the choice of time bounds. The threshold $T_i$ is given in the theorem. Given $T_i$, we choose
a suitable threshold  $t'_i$ (the gap between $t'_i$ and $T_i$ must be large enough, so that the argument on Stage~3 works).
Then, given $t_i'$ we choose a suitable value $t_i$  (the gap between $t_i$ and $t'_i$ must be large enough, so that the proof of Lemma~\ref{l:1} is valid). 
The choice of $t_i$ determines the value of $h_i$ on Stage~1.   The definition of these sequences is inductive: to define
$t_i$ we need to know $t_{i-1}$ and the  time required to produce the standard blocks of level $(i-1)$ (see the proof of Lemma~\ref{l:1}).

As $T_i$ is a computable function of $i$, we can choose computable  sequences $t''_i$, $t'_i$, $t_i$, and $h_i$. This observation concludes the proof.

\section{Proof of Lemma~\ref{l:kass-madden}}\label{appendix-kass-madden}

\begin{figure}
\begin{center}
  \begin{tikzpicture}[scale=0.27,x=1cm,baseline=2.125cm]

    \foreach \x in {-22,...,8} \foreach \y in {8,...,15}
    {
          \path[fill=black!7,draw=black] (\x,\y) rectangle ++ (1,1);
    }

    \foreach \x in {-22,...,8} \foreach \y in {16,...,23}
    {
          \path[fill=black!7,draw=black] (\x,\y) rectangle ++ (1,1);
    }

    \foreach \x in {-22,...,-1} \foreach \y in {0,...,7}
    {
          \path[fill=black!7,draw=black] (\x,\y) rectangle ++ (1,1);
    }

     
    \foreach \x in {0,...,8} \foreach \y in {0,...,7}
    {
          \path[fill=black!7,draw=black] (\x,\y) rectangle ++ (1,1);
    }

 \edef\y{0}
 \edef\xmax{3}
 \foreach \x in {0,...,\xmax}
 \path[fill=black!56,draw=black] (\x,\y) rectangle ++ (1,1);

 \edef\xmin{-19}  
 \foreach \x in {\xmin,...,-1} 
 \path[fill=black!66,draw=black] (\x,\y) rectangle ++ (1,1);
 \edef\xmaxtop{4}  
 \edef\ytop{23}  
 \foreach \x in {\xmin,...,\xmaxtop} 
 \path[fill=red!41,draw=black] (\x,\ytop) rectangle ++ (1,1);

 \edef\y{1}
 \edef\xmax{2}
 \foreach \x in {0,...,\xmax}
 \path[fill=black!56,draw=black] (\x,\y) rectangle ++ (1,1);

 \edef\xmin{-18}  
 \foreach \x in {\xmin,...,-1} 
 \path[fill=black!66,draw=black] (\x,\y) rectangle ++ (1,1);
 \edef\xmaxtop{3}  
 \edef\ytop{22}  
 \foreach \x in {\xmin,...,\xmaxtop} 
 \path[fill=red!41,draw=black] (\x,\ytop) rectangle ++ (1,1);

 \edef\y{2}
 \edef\xmax{7}
 \foreach \x in {0,...,\xmax}
 \path[fill=black!56,draw=black] (\x,\y) rectangle ++ (1,1);

\edef\xmin{-11}  
 \foreach \x in {\xmin,...,-1} 
 \path[fill=black!66,draw=black] (\x,\y) rectangle ++ (1,1);
 \edef\xmaxtop{8}  
 \edef\ytop{21}  
 \foreach \x in {\xmin,...,\xmaxtop} 
 \path[fill=red!41,draw=black] (\x,\ytop) rectangle ++ (1,1);

 \edef\y{3}
 \edef\xmax{4}
 \foreach \x in {0,...,\xmax}
 \path[fill=black!56,draw=black] (\x,\y) rectangle ++ (1,1);

\edef\xmin{-12}  
 \foreach \x in {\xmin,...,-1} 
 \path[fill=black!66,draw=black] (\x,\y) rectangle ++ (1,1);
 \edef\xmaxtop{5}  
 \edef\ytop{20}  
 \foreach \x in {\xmin,...,\xmaxtop} 
 \path[fill=red!41,draw=black] (\x,\ytop) rectangle ++ (1,1);

 \edef\y{4}
 \edef\xmax{3}
 \foreach \x in {0,...,\xmax}
 \path[fill=black!56,draw=black] (\x,\y) rectangle ++ (1,1);

\edef\xmin{-11}  
 \foreach \x in {\xmin,...,-1} 
 \path[fill=black!66,draw=black] (\x,\y) rectangle ++ (1,1);
 \edef\xmaxtop{4}  
 \edef\ytop{19}  
 \foreach \x in {\xmin,...,\xmaxtop} 
 \path[fill=red!41,draw=black] (\x,\ytop) rectangle ++ (1,1);

 \edef\y{5}
 \edef\xmax{1}
 \foreach \x in {0,...,\xmax}
 \path[fill=black!56,draw=black] (\x,\y) rectangle ++ (1,1);

\edef\xmin{-11}  
 \foreach \x in {\xmin,...,-1} 
 \path[fill=black!66,draw=black] (\x,\y) rectangle ++ (1,1);
 \edef\xmaxtop{2}  
 \edef\ytop{18}  
 \foreach \x in {\xmin,...,\xmaxtop} 
 \path[fill=red!41,draw=black] (\x,\ytop) rectangle ++ (1,1);

 \edef\y{6}
 \edef\xmax{3}
 \foreach \x in {0,...,\xmax}
 \path[fill=black!56,draw=black] (\x,\y) rectangle ++ (1,1);

\edef\xmin{-7}  
 \foreach \x in {\xmin,...,-1} 
 \path[fill=black!66,draw=black] (\x,\y) rectangle ++ (1,1);
 \edef\xmaxtop{4}  
 \edef\ytop{17}  
 \foreach \x in {\xmin,...,\xmaxtop} 
 \path[fill=red!41,draw=black] (\x,\ytop) rectangle ++ (1,1);

 \edef\y{7}
 \edef\xmax{5}
 \foreach \x in {0,...,\xmax}
 \path[fill=black!56,draw=black] (\x,\y) rectangle ++ (1,1);

\edef\xmin{-3}  
 \foreach \x in {\xmin,...,-1} 
 \path[fill=black!66,draw=black] (\x,\y) rectangle ++ (1,1);
 \edef\xmaxtop{6}  
 \edef\ytop{16}  
 \foreach \x in {\xmin,...,\xmaxtop} 
 \path[fill=red!41,draw=black] (\x,\ytop) rectangle ++ (1,1);


 \path[draw=black,line width=0.7mm] (0,0) rectangle ++ (8,8);



 \path[draw=blue,line width=0.9mm,dashed] (-12,3) rectangle ++ (18,18);
 \path[fill=black!7,draw=black,line width=0.6mm] (5,3) rectangle ++ (1,1);

 \draw[-latex,ultra thick,blue!80](12,10)node[right]{due to the bar of red letters}
       to[out=-90,in=0] (6.2,3.5);
\node[right,blue!80] at (12,8.5){on the top, this row in the };
\node[right,blue!80] at (12,7.0){$8\times 8$ frame can start with at};
\node[right,blue!80] at (12,5.5){most $5$ black letters};

 \draw[-latex,ultra thick,black!95](15,-6)node[below]{we control  $\max$ of ${\cal E}$ for this $n\times n$ frame}
       to[out=180,in=-90] (5.5,1.5);

\node[right,black!95] at (-27.3,0.3){\tiny line $1$};
\node[right,black!95] at (-27.3,1.3){\tiny line $2$};
\node[right,black!95] at (-27.3,2.3){\tiny line $3$};

\node[right,black!95] at (-27.3,5.3){\tiny  $\vdots$};

\node[right,black!95] at (-27.3,7.3){\tiny line $n$};

\node[right,black!95] at (-27.3,16.3){\tiny line $2n+1$};
\node[right,black!95] at (-27.3,17.3){\tiny line $2n+2$};

\node[right,black!95] at (-27.3,20.7){\tiny  $\vdots$};

\node[right,black!95] at (-27.3,23.3){\tiny line $3n$};

\draw [decorate,decoration={brace,amplitude=10pt},xshift=0pt,yshift=4pt,thick] (-19,24.3) -- (5,24.3) node [black,midway,yshift=17.0] {\footnotesize $3n$};

\draw [decorate,decoration={brace,amplitude=10pt},xshift=0pt,yshift=4pt,thick] (4,-0.5) -- (-19,-0.5) node [black,midway,yshift=-17.0] {\footnotesize $3n-1$};

\end{tikzpicture}
\end{center}
\caption{An $n\times n$ pattern $P$ with a neighborhood that enforces the desired maximum of~${\cal E}$.} \label{f:012-enforcing-standard}
\vspace{-5pt}
\end{figure}

In this section we sketch the proof of Lemma~\ref{l:kass-madden}. For every \emph{simple pattern} $P$ of size $n\times n$ we should construct a configuration $R$ on the complement of $B_n$, so that 
 \begin{itemize}
 \item[(i)] $P$ and $R$ are compatible,
 \item[(ii)] for every other simple pattern $P'$ compatible with $R$ we have ${\cal E}(P')\preccurlyeq {\cal E}(P)$.
 \end{itemize}
To build the required $R$, we follow the construction from \cite{kass-madden}.

By definition, each row of $P$ consists of a contiguous sequence of  black letters followed by a contiguous sequence of  white letters, as shown in Fig.~\ref{f:012-standard}. The  pattern $R$ will consist of a finite number of black and red letters (the other letters will be white).

\textbf{Black letters in $R$.}
To construct $R$, we extend each stripes of black letters in $P$ to the left, so that in the first line we get a contiguous sequence of $(3n-1)$ black letters (including those black letters that belong to $P$), in the second line  a contiguous sequence of $(3n-3)$  black letters, in the third line  a contiguous sequence of $(3n-5)$  black letters, etc. In the $n$-th line we obtain a contiguous sequence of  $(n+1)$  black letter, see Fig.~\ref{f:012-enforcing-standard}.

\textbf{Red letters in $R$.} Similarly, we put in $R$ stripes of  red letters: $3n$ contiguous red letters in line $3n$,  $(3n-2)$ contiguous red letters in line $3n-1$,
 \ldots, $(n+2)$ contiguous red letters in line $(2n+1)$. We place these stripes of red letters so that for each $i=1,\ldots,n$  the leftmost red letter in the line $(3n-i+1)$ is vertically  aligned  with the leftmost black letter in the line $i$, as shown in Fig.~\ref{f:012-enforcing-standard}. 

All other letters outside $B_n$ are made white.

\smallskip

\emph{Claim~1.}  The constructed $R$ is compatible with $P$. 

\smallskip

\emph{Proof of Claim~1:}
This fact is easy to verify: we have chosen the lengths of black and red stripes so that they cannot form a forbidden pattern (as in Fig.~\ref{f:012-forbidden}), regardless the horizontal placement of each stripes. Indeed, on the one hand, the black letters of the $i$-th line cannot interfere with the red stripes in lines $3n, 3n-1, \ldots, 3n-i$, since \emph{this black stripe} is too short to form a forbidden pattern together with any of these red stripes; on the other hand, the black letters of the $i$-th line cannot interfere 
with the red stripes in lines $3n-i-1, 3n-i-2, \ldots, 2n+1$, since \emph{those red stripes} are too short. \qed

\smallskip

\emph{Claim ~2.}  The constructed $R$ is compatible only with simple patterns $P'$ such that  ${\cal E}(P')\preccurlyeq {\cal E}(P)$.

\smallskip


%
\begin{figure}[H]
\begin{center}
  \begin{tikzpicture}[scale=0.24,x=1cm,baseline=2.125cm]

    \foreach \x in {-22,...,8} \foreach \y in {8,...,15}
    {
          \path[fill=black!7,draw=black] (\x,\y) rectangle ++ (1,1);
    }

    \foreach \x in {-22,...,8} \foreach \y in {16,...,23}
    {
          \path[fill=black!7,draw=black] (\x,\y) rectangle ++ (1,1);
    }

    \foreach \x in {-22,...,-1} \foreach \y in {0,...,7}
    {
          \path[fill=black!7,draw=black] (\x,\y) rectangle ++ (1,1);
    }

     
    \foreach \x in {0,...,8} \foreach \y in {0,...,7}
    {
          \path[fill=black!7,draw=black] (\x,\y) rectangle ++ (1,1);
    }

 \edef\y{0}
 \edef\xmax{3}
 \foreach \x in {0,...,\xmax}
 \path[fill=black!56,draw=black] (\x,\y) rectangle ++ (1,1);

 \edef\xmin{-19}  
 \foreach \x in {\xmin,...,-1} 
 \path[fill=black!66,draw=black] (\x,\y) rectangle ++ (1,1);
 \edef\xmaxtop{4}  
 \edef\ytop{23}  
 \foreach \x in {\xmin,...,\xmaxtop} 
 \path[fill=red!41,draw=black] (\x,\ytop) rectangle ++ (1,1);

 \edef\y{1}
 \edef\xmax{1}
 \foreach \x in {0,...,\xmax}
 \path[fill=black!56,draw=black] (\x,\y) rectangle ++ (1,1);
  \path[draw=black,line width=0.5mm,dashed] (0,\y) rectangle ++ (3,1);

 \edef\xmin{-18}  
 \foreach \x in {\xmin,...,-1} 
 \path[fill=black!66,draw=black] (\x,\y) rectangle ++ (1,1);
 \edef\xmaxtop{3}  
 \edef\ytop{22}  
 \foreach \x in {\xmin,...,\xmaxtop} 
 \path[fill=red!41,draw=black] (\x,\ytop) rectangle ++ (1,1);

 \edef\y{2}
 \edef\xmax{2}
 \foreach \x in {0,...,\xmax}
 \path[fill=black!56,draw=black] (\x,\y) rectangle ++ (1,1);
  \path[draw=black,line width=0.5mm,dashed] (0,\y) rectangle ++ (8,1);

\edef\xmin{-11}  
 \foreach \x in {\xmin,...,-1} 
 \path[fill=black!66,draw=black] (\x,\y) rectangle ++ (1,1);
 \edef\xmaxtop{8}  
 \edef\ytop{21}  
 \foreach \x in {\xmin,...,\xmaxtop} 
 \path[fill=red!41,draw=black] (\x,\ytop) rectangle ++ (1,1);

 \edef\y{3}
 \edef\xmax{4}
 \foreach \x in {0,...,\xmax}
 \path[fill=black!56,draw=black] (\x,\y) rectangle ++ (1,1);

\edef\xmin{-12}  
 \foreach \x in {\xmin,...,-1} 
 \path[fill=black!66,draw=black] (\x,\y) rectangle ++ (1,1);
 \edef\xmaxtop{5}  
 \edef\ytop{20}  
 \foreach \x in {\xmin,...,\xmaxtop} 
 \path[fill=red!41,draw=black] (\x,\ytop) rectangle ++ (1,1);

 \edef\y{4}
 \edef\xmax{1}
 \foreach \x in {0,...,\xmax}
 \path[fill=black!56,draw=black] (\x,\y) rectangle ++ (1,1);
  \path[draw=black,line width=0.5mm,dashed] (0,\y) rectangle ++ (4,1);

\edef\xmin{-11}  
 \foreach \x in {\xmin,...,-1} 
 \path[fill=black!66,draw=black] (\x,\y) rectangle ++ (1,1);
 \edef\xmaxtop{4}  
 \edef\ytop{19}  
 \foreach \x in {\xmin,...,\xmaxtop} 
 \path[fill=red!41,draw=black] (\x,\ytop) rectangle ++ (1,1);

 \edef\y{5}
 \edef\xmax{1}
 \foreach \x in {0,...,\xmax}
 \path[fill=black!56,draw=black] (\x,\y) rectangle ++ (1,1);

\edef\xmin{-11}  
 \foreach \x in {\xmin,...,-1} 
 \path[fill=black!66,draw=black] (\x,\y) rectangle ++ (1,1);
 \edef\xmaxtop{2}  
 \edef\ytop{18}  
 \foreach \x in {\xmin,...,\xmaxtop} 
 \path[fill=red!41,draw=black] (\x,\ytop) rectangle ++ (1,1);

 \edef\y{6}
 \edef\xmax{3}
 \foreach \x in {0,...,\xmax}
 \path[fill=black!56,draw=black] (\x,\y) rectangle ++ (1,1);

\edef\xmin{-7}  
 \foreach \x in {\xmin,...,-1} 
 \path[fill=black!66,draw=black] (\x,\y) rectangle ++ (1,1);
 \edef\xmaxtop{4}  
 \edef\ytop{17}  
 \foreach \x in {\xmin,...,\xmaxtop} 
 \path[fill=red!41,draw=black] (\x,\ytop) rectangle ++ (1,1);

 \edef\y{7}
 \edef\xmax{3}
 \foreach \x in {0,...,\xmax}
 \path[fill=black!56,draw=black] (\x,\y) rectangle ++ (1,1);
  \path[draw=black,line width=0.5mm,dashed] (0,\y) rectangle ++ (6,1);

\edef\xmin{-3}  
 \foreach \x in {\xmin,...,-1} 
 \path[fill=black!66,draw=black] (\x,\y) rectangle ++ (1,1);
 \edef\xmaxtop{6}  
 \edef\ytop{16}  
 \foreach \x in {\xmin,...,\xmaxtop} 
 \path[fill=red!41,draw=black] (\x,\ytop) rectangle ++ (1,1);


 \path[draw=black,line width=0.7mm] (0,0) rectangle ++ (8,8);



 \draw[-latex,ultra thick,black!95](15,-6)node[below]{this $n\times n$ pattern $P'$ is compatible with the neighborhood}
       to[out=180,in=-90] (5.5,1.5);

\end{tikzpicture}
\end{center}
\caption{A pattern $P'$ with ${\cal E}(P')\preccurlyeq {\cal E}(P)$ matches the neighborhood.} \label{f:012-enforcing-ex1}
\vspace{10pt}
\end{figure}

\bigskip

\begin{figure}[H]
\vspace{-5pt}
\begin{center}
  \begin{tikzpicture}[scale=0.24,x=1cm,baseline=2.125cm]

    \foreach \x in {-22,...,8} \foreach \y in {8,...,15}
    {
          \path[fill=black!7,draw=black] (\x,\y) rectangle ++ (1,1);
    }

    \foreach \x in {-22,...,8} \foreach \y in {16,...,23}
    {
          \path[fill=black!7,draw=black] (\x,\y) rectangle ++ (1,1);
    }

    \foreach \x in {-22,...,-1} \foreach \y in {0,...,7}
    {
          \path[fill=black!7,draw=black] (\x,\y) rectangle ++ (1,1);
    }

     
    \foreach \x in {0,...,8} \foreach \y in {0,...,7}
    {
          \path[fill=black!7,draw=black] (\x,\y) rectangle ++ (1,1);
    }

 \edef\y{0}
 \edef\xmax{3}
 \foreach \x in {0,...,\xmax}
 \path[fill=black!56,draw=black] (\x,\y) rectangle ++ (1,1);

 \edef\xmin{-19}  
 \foreach \x in {\xmin,...,-1} 
 \path[fill=black!66,draw=black] (\x,\y) rectangle ++ (1,1);
 \edef\xmaxtop{4}  
 \edef\ytop{23}  
 \foreach \x in {\xmin,...,\xmaxtop} 
 \path[fill=red!41,draw=black] (\x,\ytop) rectangle ++ (1,1);

 \edef\y{1}
 \edef\xmax{2}
 \foreach \x in {0,...,\xmax}
 \path[fill=black!56,draw=black] (\x,\y) rectangle ++ (1,1);

 \edef\xmin{-18}  
 \foreach \x in {\xmin,...,-1} 
 \path[fill=black!66,draw=black] (\x,\y) rectangle ++ (1,1);
 \edef\xmaxtop{3}  
 \edef\ytop{22}  
 \foreach \x in {\xmin,...,\xmaxtop} 
 \path[fill=red!41,draw=black] (\x,\ytop) rectangle ++ (1,1);

 \edef\y{2}
 \edef\xmax{7}
 \foreach \x in {0,...,\xmax}
 \path[fill=black!56,draw=black] (\x,\y) rectangle ++ (1,1);

\edef\xmin{-11}  
 \foreach \x in {\xmin,...,-1} 
 \path[fill=black!66,draw=black] (\x,\y) rectangle ++ (1,1);
 \edef\xmaxtop{8}  
 \edef\ytop{21}  
 \foreach \x in {\xmin,...,\xmaxtop} 
 \path[fill=red!41,draw=black] (\x,\ytop) rectangle ++ (1,1);

 \edef\y{3}
 \edef\xmax{4}
 \foreach \x in {0,...,\xmax}
 \path[fill=black!56,draw=black] (\x,\y) rectangle ++ (1,1);

\edef\xmin{-12}  
 \foreach \x in {\xmin,...,-1} 
 \path[fill=black!66,draw=black] (\x,\y) rectangle ++ (1,1);
 \edef\xmaxtop{5}  
 \edef\ytop{20}  
 \foreach \x in {\xmin,...,\xmaxtop} 
 \path[fill=red!41,draw=black] (\x,\ytop) rectangle ++ (1,1);

 \edef\y{4}
 \edef\xmax{3}
 \foreach \x in {0,...,\xmax}
 \path[fill=black!56,draw=black] (\x,\y) rectangle ++ (1,1);

\edef\xmin{-11}  
 \foreach \x in {\xmin,...,-1} 
 \path[fill=black!66,draw=black] (\x,\y) rectangle ++ (1,1);
 \edef\xmaxtop{4}  
 \edef\ytop{19}  
 \foreach \x in {\xmin,...,\xmaxtop} 
 \path[fill=red!41,draw=black] (\x,\ytop) rectangle ++ (1,1);

 \edef\y{5}
 \edef\xmax{1}
 \foreach \x in {0,...,\xmax}
 \path[fill=black!56,draw=black] (\x,\y) rectangle ++ (1,1);

\edef\xmin{-11}  
 \foreach \x in {\xmin,...,-1} 
 \path[fill=black!66,draw=black] (\x,\y) rectangle ++ (1,1);
 \edef\xmaxtop{2}  
 \edef\ytop{18}  
 \foreach \x in {\xmin,...,\xmaxtop} 
 \path[fill=red!41,draw=black] (\x,\ytop) rectangle ++ (1,1);

 \edef\y{6}
 \edef\xmax{3}
 \foreach \x in {0,...,\xmax}
 \path[fill=black!56,draw=black] (\x,\y) rectangle ++ (1,1);

\edef\xmin{-7}  
 \foreach \x in {\xmin,...,-1} 
 \path[fill=black!66,draw=black] (\x,\y) rectangle ++ (1,1);
 \edef\xmaxtop{4}  
 \edef\ytop{17}  
 \foreach \x in {\xmin,...,\xmaxtop} 
 \path[fill=red!41,draw=black] (\x,\ytop) rectangle ++ (1,1);

 \edef\y{7}
 \edef\xmax{5}
 \foreach \x in {0,...,\xmax}
 \path[fill=black!56,draw=black] (\x,\y) rectangle ++ (1,1);

\edef\xmin{-3}  
 \foreach \x in {\xmin,...,-1} 
 \path[fill=black!66,draw=black] (\x,\y) rectangle ++ (1,1);
 \edef\xmaxtop{6}  
 \edef\ytop{16}  
 \foreach \x in {\xmin,...,\xmaxtop} 
 \path[fill=red!41,draw=black] (\x,\ytop) rectangle ++ (1,1);


 \path[draw=black,line width=0.7mm] (0,0) rectangle ++ (8,8);



 \path[draw=blue,line width=0.9mm,dashed] (-12,3) rectangle ++ (18,18);
 \path[fill=black!66,draw=black,line width=0.4mm] (5,3) rectangle ++ (1,1);

 \draw[-latex,ultra thick,blue!80](12,8)node[right]{by adding one supplementary}
       to[out=-90,in=0] (5.5,3.5);
\node[right,blue!80] at (12,6.5){ black letter we get a forbidden};
\node[right,blue!80] at (12,5.0){pattern };

 \draw[-latex,ultra thick,black!95](15,-6)node[below]{this $n\times n$ pattern $P''$ is  incompatible with the neighborhood}
       to[out=180,in=-90] (5.5,1.5);

\end{tikzpicture}
\end{center}
\caption{A pattern $P''$ with ${\cal E}(P'')\not\preccurlyeq {\cal E}(P)$ does not match the neighborhood.} \label{f:012-enforcing-ex2}
\vspace{-5pt}
\end{figure}

\emph{Proof of Claim~2:}
If $R$ is compatible with an $n\times n$ pattern $P'$, the profile of  $P'$ is not determined uniquely. 
In fact, $R$ can be  compatible with simple patterns $P'$ 
whose profiles are \emph{strictly less} than the profile of $P$ (in each row of $P'$ the number of black letters must be not greater than the number of black letters in the corresponding row of $P$), see Fig.~\ref{f:012-enforcing-ex1}. On the other hand, if at least one row of $P'$ contains more black letters that the same row in $P$, than $P'$ and $R$ are incompatible, i.e., the joint of $P'$ and $R$ contains   a forbidden pattern, as shown in  Fig.~\ref{f:012-enforcing-ex2}.
\qed

\smallskip

The lemma follows from Claim~1 and Claim~2. For a more detailed argument we refer the reader to~\cite{kass-madden}.

\begin{remark}
In the construction discussed above, pattern $R$ does not determine uniquely the epitomes of $P'$ compatible with $R$ (these epitomes can be different, though they all must be \emph{not greater} than the epitome of the initial pattern $P$). This is why we cannot apply Proposition~\ref{p:epitomes}, and we have to use  Proposition~\ref{pbis:epitomes}.
\end{remark}

\end{document}